\pdfoutput=1
\documentclass[11pt]{article}

\usepackage[letterpaper,margin=0.9in]{geometry}
\usepackage{bm, amsmath, amssymb, amsthm, amsfonts}
\usepackage{hyperref}
\usepackage[capitalize,nameinlink]{cleveref} 
\usepackage{cite}
\usepackage{framed}
\usepackage[framemethod=tikz]{mdframed}
\usepackage{appendix}
\usepackage{graphicx}
\usepackage{color}
\usepackage{varwidth}
\usepackage{wrapfig}
\usepackage{algorithm2e}
\usepackage[noend]{algpseudocode}
\usepackage[textsize=tiny]{todonotes}

\usepackage{enumitem}
\setitemize{noitemsep,topsep=3pt,parsep=3pt,partopsep=3pt}

\usepackage{xspace}
\usepackage{xifthen}

\definecolor{darkgreen}{rgb}{0,0.5,0}
\definecolor{darkblue}{rgb}{0,0,0.8}
\usepackage{hyperref}

\hypersetup{
	colorlinks=true,                  
	linkcolor=darkblue,
	citecolor=darkgreen
}

\newtheorem{theorem}{Theorem}[section]
\newtheorem{lemma}[theorem]{Lemma}
\newtheorem{corollary}[theorem]{Corollary}
\newtheorem{observation}[theorem]{Observation}
\newtheorem{definition}{Definition}[section]
\newtheorem*{remark}{Remark}

\renewcommand{\vec}[1]{\ensuremath{\boldsymbol{#1}}}
\newcommand{\whp}{w.h.p.}
\newcommand{\multicoloring}{multi-coloring}
\newcommand{\nonmulticoloring}{non-multi coloring}

\newcommand{\calT}{\ensuremath{\mathcal{T}}}

\newcommand{\calG}{\ensuremath{\mathcal{G}}}
\newcommand{\calE}{\ensuremath{\mathcal{E}}}
\newcommand{\calV}{\ensuremath{\mathcal{V}}}
\newcommand{\calI}{\ensuremath{\mathcal{I}}}
\newcommand{\calA}{\ensuremath{\mathcal{A}}}
\newcommand{\calP}{\ensuremath{\mathcal{P}}}
\newcommand{\calC}{\ensuremath{\mathcal{C}}}

\newcommand{\instances}{\mathcal{C}}

\newcommand{\ignore}[1]{}

\algnewcommand\algorithmicswitch{\textbf{switch}}
\algnewcommand\algorithmiccase{\textbf{case}}

\algdef{SE}[SWITCH]{Switch}{EndSwitch}[1]{\algorithmicswitch\ #1\ \algorithmicdo}{\algorithmicend\ \algorithmicswitch}%
\algdef{SE}[CASE]{Case}{EndCase}[1]{\algorithmiccase\ #1}{\algorithmicend\ \algorithmiccase}%
\algtext*{EndSwitch}%
\algtext*{EndCase}%

\newcommand{\LOCAL}{\ensuremath{\mathsf{LOCAL}}\xspace}
\newcommand{\SLOCAL}{\ensuremath{\mathsf{SLOCAL}}\xspace}

\newcommand{\eps}{\varepsilon}
\newcommand{\poly}{\operatorname{\text{{\rm poly}}}}

\newcommand{\set}[1]{\left\{#1\right\}}
\newcommand{\logStar}[1]{\log^{*} #1}

\DeclareMathOperator{\polylog}{\poly\log}

\newcommand{\N}{\mathbb{N}}
\newcommand{\ZZ}{\mathbb{Z}}
\newcommand{\bigO}{O}

\newcommand{\Alg}{\mathcal{A}}

\newcommand{\complexityclass}[2][]{\ensuremath{\mathsf{#2}\ifthenelse{\isempty{#1}}{}{(#1)}}}
\newcommand{\complexClass}{\complexityclass[]{C}}
\newcommand{\loc}[1][]{\complexityclass[#1]{LOCAL}}
\newcommand{\ld}[1][]{\complexityclass[#1]{LD}}
\newcommand{\seqloc}[1][]{\complexityclass[#1]{SLOCAL}}
\newcommand{\randloc}[2][]{\complexityclass[#1]{RLOCAL_{#2}}}
\newcommand{\randseqloc}[2][]{\complexityclass[#1]{RSLOCAL_{#2}}}
\newcommand{\polyloc}{\complexityclass{P}\text{-}\complexityclass{LOCAL}}
\newcommand{\polyseqloc}{\complexityclass{P}\text{-}\complexityclass{SLOCAL}}
\newcommand{\polyrandloc}[1]{\complexityclass{P}\text{-}\complexityclass{RLOCAL_{#1}}}
\newcommand{\polyrandseqloc}[1]{\complexityclass{P}\text{-}\complexityclass{RSLOCAL_{#1}}}

\newcommand{\ym}[1]{\textit{ \textcolor[rgb]{1,0,0}{[YM: #1]}}}

\newcommand{\hide}[1]{}

\newcommand{\FullOrShort}{short}

\ifthenelse{\equal{\FullOrShort}{full}}{
	
  \newcommand{\fullOnly}[1]{#1}
  \newcommand{\shortOnly}[1]{}

}{

  \newcommand{\shortOnly}[1]{#1}
  \newcommand{\fullOnly}[1]{}

}
  
\begin{document}

\title{On the Complexity of Local Distributed Graph
  Problems}

\author{
Mohsen Ghaffari \\
ETH Zurich, Switzerland \\
\small \tt ghaffari@inf.ethz.ch
\and
Fabian Kuhn\thanks{Supported by ERC
  Grant No.\ 336495 (ACDC)}\\
University of Freiburg, Germany \\
\small\tt kuhn@cs.uni-freiburg.de
\and 
Yannic Maus\footnotemark[1]\\
University of Freiburg, Germany \\
\small\tt yannic.maus@cs.uni-freiburg.de
}

\date{}

\maketitle

\thispagestyle{empty}

\begin{abstract}
  This paper is centered on the complexity of graph problems in the
  well-studied \LOCAL model of distributed computing, introduced by
  Linial [FOCS '87].  It is widely known that for many of the classic
  distributed graph problems (including maximal independent set (MIS)
  and $(\Delta+1)$-vertex coloring), the randomized complexity is at
  most polylogarithmic in the size $n$ of the network, while the best
  deterministic complexity is typically $2^{O(\sqrt{\log n})}$.
  Understanding and potentially narrowing down this exponential gap is
  considered to be one of the central long-standing open questions in
  the area of distributed graph algorithms.

  We investigate the problem by introducing a complexity-theoretic
  framework that allows us to shed some light on the role of
  randomness in the \LOCAL model. We define the \SLOCAL model as a
  sequential version of the \LOCAL model. Our framework allows us to
  prove \emph{completeness} results with respect to the class of
  problems which can be solved efficiently in the \SLOCAL model,
  implying that if any of the complete problems can be solved
  deterministically in $\poly\log n$ rounds in the \LOCAL model, we
  can deterministically solve all efficient \SLOCAL-problems
  (including MIS and $(\Delta+1)$-coloring) in $\poly\log n$ rounds in
  the \LOCAL model.

  Perhaps most surprisingly, we show that a rather rudimentary looking
  graph coloring problem is \emph{complete} in the above sense: Color
  the nodes of a graph with colors red and blue such that each node of
  sufficiently large polylogarithmic degree has at least one neighbor
  of each color. The problem admits a trivial zero-round randomized
  solution.  The result can be viewed as showing that the only
  obstacle to getting efficient determinstic algorithms in the \LOCAL
  model is an efficient algorithm to \emph{approximately round
    fractional values into integer values}.

  In addition, our formal framework also allows us to develop
  polylogarithmic-time randomized distributed algorithms in a simpler
  way. As a result, we provide a polylog-time distributed
  approximation scheme for arbitrary distributed covering and packing
  integer linear programs.
\end{abstract}


\newpage

\hide{
\thispagestyle{empty}
\tableofcontents

\newpage
}

\setcounter{page}{1}

\section{Introduction \& Related Work}
\label{sec:intro}
The question of whether a given distributed problem can be solved
locally has been at the center of the theory of distributed graph
algorithms since the 1980s, especially starting with the seminal
work of Awerbuch, Goldberg, Luby, and Plotkin~\cite{awerbuch89},
Linial~\cite{linial92}, and Naor and Stockmeyer~\cite{naor95}. The
locality of distributed computations is captured by the \LOCAL
model~\cite{linial92,peleg00}, defined as follows: a network is modeled as an
undirected graph $G=(V,E)$, the nodes $V$ are the network devices, and
the edges $E$ are bidirectional communication links. Time is divided
into synchronous communication rounds. In each round, each node can
perform some arbitrary internal computation, send a message of
possibly arbitrary size to each of its neighbors, and receive the
messages sent to it by its neighbors.  A typical objective in this
setting is to solve some given graph problem on the network $G$ by a
distributed algorithm. For example, classic problems include computing a vertex or an edge
coloring with a given number of colors
\cite{awerbuch89,barenboim10,barenboim12,barenboimelkin_book,BEK15,barenboim15,chang16,cole86,fraigniaud16,goldberg88,linial92,disc16_coloring,hsinhao_coloring,szegedy93},
computing a maximal independent set (MIS) or a maximal
matching \cite{alon86,barenboim12,hanckowiak01,kuhn16_jacm,luby86,linial92,ghaffari16_MIS},
or approximating classic optimization problems with
local constraints such as maximum matching, minimum vertex cover, or
minimum dominating set
\cite{czygrinow04,dubhashi05,goeoes14_DISTCOMP,jia02,nearsighted,kuhn16_jacm,suomela_survey}.
In any $r$-round algorithm in the \LOCAL model, the output of a node
$v$ can depend only on the initial states of nodes in the $r$-hop
neighborhood of $v$, but it can be an arbitrary function of this
neighborhood \cite{linial92}. Therefore, the \LOCAL model captures a
core issue of distributed computations in a precise mathematical
sense: What global goals can be achieved based on only local
information.

\medskip

\noindent\textbf{The Role of Randomness:} 
A major challenge in designing fast distributed algorithms in the
\LOCAL model is to break symmetries and coordinate actions among
nearby nodes.
It is maybe not surprising that this has turned out much easier if the
nodes are allowed to use randomization.\footnote{For example when
  computing a coloring with $\Delta+1$ colors (where $\Delta$ is the
  maximum degree of the network graph $G$), with high probability, it
  suffices to iterate the following simple randomized coloring scheme
  $O(\log n)$ times: Given any partial initial coloring, each
  uncolored node $v$ picks a uniformly random color among the colors
  still available to $v$. If $v$ randomly picks a color $x$ not chosen
  by any neighbor in the same iteration, $v$ outputs color $x$ and
  otherwise the color of $v$ remains undecided.} As a result, for many
important problems, there currently is an \emph{exponential gap}
between the time complexity of the best \emph{randomized} and the best
\emph{deterministic} distributed algorithms. Typically, an algorithm
in the \LOCAL model is considered efficient if its time complexity is
polylogarithmic in the number of nodes $n$. For a large number of
fundamental distributed graph problems (including MIS and
$(\Delta+1)$-coloring), there are logarithmic or polylog-time
randomized distributed algorithms (e.g.,
\cite{alon86,ghaffari16_MIS,nearsighted,linial92,linial93,luby86,hsinhao_coloring}),
whereas the best kown deterministic distributed algorithms have time
complexity $2^{O(\sqrt{\log n})}$
\cite{awerbuch89,panconesi95}. Understanding whether this exponential
separation is inherent is considered to be one of the major
long-standing open problems of the area
\cite{barenboimelkin_book,linial92}. Recently, in \cite{chang16, ghaffari17} (see
also \cite{LLL_lowerbound}), it has been shown that in the
\LOCAL model, there are problems --- e.g., $\Delta$-coloring trees or
computing a sinkless orientation --- with a deterministic complexity of
$\Theta(\log_\Delta n)$, while the randomized complexity is $\Theta(\log\log_\Delta n)$. However, the classic open question of
whether such an exponential separation also holds when ignoring
polylogarithmic factors remains open. One of the main objectives of
our work is to shed some light on this long-standing open problem.

\medskip

\noindent\textbf{A Complexity-Theoretic Perspective:} In this paper, we investigate the \emph{role of randomness} in
distributed graph algorithms from a \emph{complexity-theoretic}
viewpoint. In particular, we study the class \polyloc\ of all graph
problems which can be solved deterministically in polylogarithmic time
in the \LOCAL model and we define a much wider class \polyseqloc\ of
problems which informally consists of all problems where the output of
all nodes is determined by sequentially looking at a polylog-radius
neighborhood of each node. In particular, the class \polyseqloc\
contains all the above mentioned classic problems for which
polylog-time randomized distributed algorithms are known and where the
current best deterministic solutions require time
$2^{O(\sqrt{\log n})}$. We prove that a number of natural
distributed graph problems are \polyseqloc-complete: If any of these
problems has a deterministic polylog-time distributed algorithm, all
problems in \polyseqloc\ can be solved deterministically in polylog
time in the \LOCAL model and thus $\polyloc=\polyseqloc$. 

Perhaps most surprisingly, we prove that the following natural and
rudimentary-looking \emph{rounding} problem is \polyseqloc-complete:
We are given a bipartite graph $B=(U\dot{\cup} V, E)$, where the
degree of each node in $U$ is at least $\log^c n$ for a desirably
large constant $c\geq 2$. The objective is to color each node in $V$
red or blue such that for each node in $U$, the degree is
approximately equally split. In fact any coarse but non-trivial
relaxation of `\emph{approximately equal}' suffices, e.g., it is
enough if the neighbors in the two colors have the same size up to
poly-logarithmic factors. Using randomization, this can be done
without any communication---i.e., in zero rounds---via independently
coloring each node in $V$ red or blue with probability $1/2$. The
problem can be seen as a basic rounding problem with linear
constraints. Hence, in a certain sense, we show that the \emph{only
  obstacle to efficient deterministic distributed algorithms} is an
\emph{efficient deterministic algorithm for rounding fractional to
  integer values}.

\medskip

\noindent\textbf{Implications on Randomized Distributed Algorithms:} From our completeness results, it also immediately follows that all
problems in \polyseqloc\ have \emph{polylog-time randomized solutions}
in the \LOCAL model. Thus, in addition to providing a tool to study
the hardness of local symmetry breaking and coordination problems, the
\polyseqloc\ model provides a useful abstraction that simplifies
studying what can be solved efficiently in the \LOCAL model when
allowing randomization. In particular, we show that computing
$(1+\eps)$-approximate solutions for general covering and packing
integer linear programs is in \polyseqloc. This directly implies that
covering and packing integer linear programs (such as e.g., the
minimum dominating set problem or the maximum independent set problem)
can be approximated arbitrarily well in polylogarithmic time in the
\LOCAL model. This significantly improves the best existing algorithms
for these problems
\cite{barenboim15_decomp,podc16_BA,jia02,kuhn16_jacm}.

\medskip

In the following, we discuss our contributions and additional related
work in more detail.

\vspace*{-3mm}

\subsection{Sequential Local Computations}
As argued, one of the main challenges in the \LOCAL model is
to locally coordinate the parallel actions of nearby nodes. Such local
coordination becomes significantly easier if we remove the inherent
parallelism of distributed computations and if the outputs of all the
nodes can be computed sequentially, one node at a time. This can be
well illustrated by the MIS or the $(\Delta+1)$-coloring problem. In
both cases, there is a trivial greedy algorithm which sequentially
processes all the nodes in an arbitrary order. In order to determine
the output value of a node $v$, the sequential MIS and
$(\Delta+1)$-coloring algorithms merely need to inspect the already
computed outputs of the neighbors of $v$.

We generalize the above basic greedy algorithms and define the \SLOCAL
model. In the \SLOCAL model, nodes are processed in an arbitrary
order. When a node $v$ is processed, it can see the current state of
its $r$-hop neighborhood for some $r\geq 0$ and compute its output as
an arbitrary function of this. In addition, $v$ can locally store an
arbitrary amount of information, which can be read by later nodes as
part of $v$'s state. We say that $r$ is the \emph{locality} of an
algorithm in the \SLOCAL model. The model is defined precisely and
discussed more thoroughly in \Cref{sec:SLOCAL}.

The \SLOCAL model is loosely related to other sequential models in
which, when studying a graph problem, the output of a single node has
to be determined by only considering a small part of the graph. In
particular, we would like to mention \emph{Local Computation Algorithms} (LCA)\cite{lca11,lca12}. In LCAs, the focus is on
bounding the local computation and the space for computing the output
of each node to a sublinear or even $\poly\log n$. In contrast, we purposefully do
not bound local computations or space in any way. As we later show
completeness w.r.t.\ complexity classes of \SLOCAL algorithms, we
would like the \SLOCAL model to be as general as possible. Unlike the
\SLOCAL model, LCAs allow some shared randomness and sometimes also
some small amount of global memory. We do not allow any globally
shared state as this would make the model too powerful\footnote{E.g., even one bit
of global memory would allow to solve \emph{leader election}, which clearly cannot be solved locally.}.



\subsection{Complexity Classes}

We introduce two basic complexity classes which are informally defined
as follows. The class $\loc(t)$ consists of all distributed graph
problems which can be solved deterministically in $t$ rounds in the
\LOCAL model. The class $\seqloc(t)$ consists of all distributed graph
problems which can be solved deterministically with locality $t$ in
the \SLOCAL model. For formal definitions of all the complexity
classes, refer to \Cref{sec:classes}. Note that the simple
greedy algorithms show that MIS and $(\Delta+1)$-coloring are in the
class $\seqloc(1)$, whereas we only know that they are in
the class $\loc\big(2^{c\sqrt{\log n}}\big)$ for some constant $c>0$
\cite{panconesi95}. We are mostly interested in \LOCAL and \SLOCAL
algorithms with locality polylogarithmic in the number of nodes $n$. Thus, we define $\polyloc := \loc\big(\log^{O(1)}n\big)$ and
$\polyseqloc := \seqloc\big(\log^{O(1)}n\big)$ to capture algorithms
with polylogarithmic locality.

Our approach can be viewed as an extension of the recent fundamental
work of Fraigniaud, Korman, and Peleg in \cite{fraigniaud13} on the
complexity of \emph{distributed decision problems}. In a distributed
decision problem, every node has to output either \emph{yes} or
\emph{no} such that for yes-instances, all nodes output \emph{yes},
whereas for no-instances, at least one node outputs \emph{no}. In
\cite{fraigniaud13}, the class $\ld(t)$ is defined as the set all
distributed decision problems which can be solved in $t$ rounds in the
\LOCAL model. The class $\loc(t)$ extends $\ld(t)$ to
\emph{distributed search problems} and we thus have
$\ld(t)\subset \loc(t)$. The work started in \cite{fraigniaud13} lead
to series of insightful results
\cite{feuilloley16_decisionhierarchy,fraigniaud14_randdecision,fraigniaud13_decisionIDs,fraigniaud15_decisionLabels}. We
would however like to stress that while in the standard sequential
setting, there are standard techniques for transforming many standard
search problems into decision problems, the situation is very
different in the distributed setting. In fact, most of the standard
distributed search problems cannot be reduced to corresponding decision versions and
studying decision problems is not sufficient to capture some of the
core difficulties when developing algorithms for the \LOCAL model.



\subsection{Problem Definitions and Completeness Results}
\label{sec:problemDefinitionResults}

We will show that all the problems in \polyseqloc\ can be solved in
randomized polylog time and in deterministic $2^{O(\sqrt{\log n})}$
time in the \LOCAL model.  Hence, except for the potential additional
power of using randomization in the \SLOCAL model, the class
(deterministic) \polyseqloc\ exactly captures what can be solved in
polylog randomized time in the \LOCAL model. To understand the
separation between randomized and deterministic distributed
algorithms, we thus need to study the deterministic complexity of the
problems in \polyseqloc\ in the \LOCAL model.

For distributed graph problems $\calP_1$ and $\calP_2$, we say that
$\calP_1$ is \emph{polylog-reducible} to $\calP_2$ if a polylog-time
deterministic distributed algorithm for $\calP_2$ implies a
polylog-time deterministic distributed algorithm for $\calP_1$. We
define a problem $\calP$ to be \emph{\polyseqloc-complete} if
$\calP\in \polyseqloc$ and any problem in $\polyseqloc$ is
polylog-reducible to $\calP$. Hence, if any \polyseqloc-complete
problem can be solved deterministically in polylog time in the \LOCAL
model, we have $\polyloc=\polyseqloc$ and thus all problems in
\polyseqloc\ have deterministic polylog-time \LOCAL algorithms.


The best known deterministic algorithms for MIS and
$(\Delta+1)$-coloring, as well as for many other problems in
\polyseqloc\ are based on a decomposition of the network into clusters
of small diameter, which was defined by Awerbuch et al.\ in
\cite{awerbuch89}.

\begin{definition}[\textbf{Network Decomposition}]\label{def:decomposition}\cite{awerbuch89}
  A weak (strong) \emph{$\big(d(n),c(n)\big)$-decomposition} of an
  $n$-node graph $G=(V,E)$ is a partition of $V$ into clusters such
  that each cluster has weak (strong) diameter at most $d(n)$ and the
  cluster graph is properly colored with colors $1,\dots,c(n)$.
\end{definition}

In \cite{awerbuch89}, it is shown that for
$d(n)=c(n)=2^{O(\sqrt{\log n\log\log n})}$, such a decomposition can
be computed deterministically in $2^{O(\sqrt{\log n\log\log n})}$
rounds in the \LOCAL model. This was later improved by Panconesi and
Srinivasan who managed to get rid of the $\log\log n$ terms in all the
above bounds \cite{panconesi95}. It is not hard to see that given a
$\big(d(n),c(n)\big)$-decomposition, an MIS, a $(\Delta+1)$-coloring,
and in fact many other standard graph problems can be computed
deterministically in time $O\big(d(n)c(n)\big)$ in the \LOCAL
model. Using the decomposition of \cite{panconesi95}, this results in
deterministic distributed algorithms with time complexity
$2^{O(\sqrt{\log n})}$.

In \cite{linial93}, Linial and Saks show that every graph has a
$\big(O(\log n),O(\log n)\big)$-decomposition and that such a
decomposition can be computed by a randomized algorithm in
$O(\log^2 n)$ rounds.\footnote{As pointed out in \cite{linial93}, the
  existence of a $\big(O(\log n),O(\log n)\big)$-decomposition
  essentially already follows implicitly from the work of Awerbuch and
  Peleg \cite{Awerbuch-Peleg1990}.} It has commonly been understood
that the network decomposition problem takes a central role in
understanding the complexity of local distributed computations
\cite{awerbuch96,awerbuch89,barenboim12_decomp,barenboim15_decomp,elkin16_decomp,linial93,panconesi95}. We
make the key significance of network decomposition formal by proving
the following theorem.

\begin{theorem}\label{thm:decomposition}
  The problem of computing a weak or strong
  $(\poly\log n, \poly\log n)$-decomposition of a given
  $n$-node network graph $G$ is \polyseqloc-complete.
\end{theorem}

Given the order $\pi$ in which an \SLOCAL-algorithm $\calA$ processes
the nodes of a graph $G$, there is a direct way to execute $\calA$ in
a distributed setting. If the locality of $\calA$ is $r$, a node $v$
can compute its output as soon as all nodes within distance $r$ which
appear before $v$ in $\pi$ have computed their outputs. If the maximum
length of such a dependency chain is $T$, this leads to a $Tr$-round
distributed algorithm for $\calA$. Unfortunately, the maximum
dependency chain cannot be bounded by a small function, e.g., if $G$
is a complete graph, there is always a dependency chain of length
$n$. However, in the \LOCAL model, for a node $v$ to determine its
output in $R$ rounds, it suffices if $v$ can learn all its dependency
chains, i.e., if all the dependency chains of $v$ are contained in the
$R$-neighborhood of $v$ in $G$. A given \SLOCAL-algorithm thus has an
efficient distributed implementation if we can find an order $\pi$ on
the nodes such that any dependency chain is contained in a
small-diameter neighborhood.

\begin{definition}[\textbf{Low Diameter Ordering}]\label{def:ordering}
  Given an $n$-node graph $G=(V,E)$, a \emph{$d(n)$-diameter ordering} of
  $G$ is an assignment of unique labels to all nodes $V$ such that
  for any path $P$ on which the labels are increasing along $P$, any
  two nodes of $P$ are within distance $d(n)$ in $G$.
\end{definition}

Note that on the complete graph, any order $\pi$ is a $1$-diameter
ordering. We will show that every $n$-node graph $G$ has an
$O(\log^2n)$-diameter ordering and that we get the following theorem.

\begin{theorem}\label{thm:ordering}
  There is a constant $c>0$ such that for every function $d(n)$ with
  $c\ln^2n\leq d(n)=\log^{O(1)}n$, computing a $d(n)$-diameter
  ordering of an $n$-node graph $G$ is \polyseqloc-complete.
\end{theorem}

Using a network decomposition or a low-diameter ordering, there is a
relatively direct way of turning a given \SLOCAL-algorithm into a
distributed one. In addition, we show \polyseqloc-completeness of the
following extremely rudimentary looking problems.

\begin{definition}[\textbf{Local Splitting}]\label{def:split}
  Given is a bipartite graph $B=(U \dot{\cup} V, E_B)$ where $E_B\subseteq
  U\times V$. For any $\lambda\in [0,1/2]$, we define a $\lambda$-local
  splitting of $B$ to be a $2$-coloring of the nodes in $V$ with
  colors red and blue such that each node $v$ has at least
  $\lfloor\lambda\cdot d(v)\rfloor$ neighbors of each color.
\end{definition}

\begin{definition}[\textbf{Weak Local Splitting}]\label{def:weaksplit}
  Given is a bipartite graph $B=(U \dot{\cup} V, E_B)$ where $E_B\subseteq
  U\times V$. We define a weak local
  splitting of $B$ to be a $2$-coloring of the nodes in $V$ with
  colors red and blue such that each node $v$ has at least $1$ neighbor of
  each color.
\end{definition}

If the minimum degree of any node in $U$ is at least $c\ln n$ for a
sufficiently large constant $c$, then $\lambda$-local splitting (even
for $\lambda$ close to $1/2$) and weak local splitting can be solved
trivially in $0$ rounds by using randomization: Color each node in $V$
independently red or blue with probability $1/2$; this coloring
satisfies the required conditions, with high probability. The
following two theorems are the main technical contribution of our
paper. They show that, in some sense, the above local splitting
problems---even the weak local splitting---already capture the core of
the difficulty in designing polylog-time deterministic \LOCAL
algorithms.

\begin{theorem}\label{thm:localsplit}
  For bipartite graphs $H=(U\dot{\cup} V, E)$ where all nodes in $U$
  have degree at least $c\ln^2 n$ for a large enough constant $c$, the
  $\lambda$-local splitting problem for any
  $\lambda=\frac{1}{\poly\log n}$ is \polyseqloc-complete.
\end{theorem}

\begin{theorem}\label{thm:weaklocalsplit}
  For bipartite graphs $H=(U\dot{\cup} V, E)$ where all nodes in $U$
  have degree $\delta/2< d(u)\leq \delta$, for any $\delta$ such that
  $c\ln^2 n \leq \delta = \log^{O(1)} n$ for a sufficiently large
  constant $c$, the weak local splitting problem is
  \polyseqloc-complete.
\end{theorem}

The local splitting problem can be viewed as a very special case of
\emph{rounding}, i.e., turning fractional values to integral values
while respecting some linear constraints: Associate a variable $x_v$
with each vertex $v\in V$ and think of each vertex $u\in U$ as two
linear constraints,
$\lambda\Delta \leq \sum_{v\in N(u)}x_v \leq (1-\lambda)\Delta$.
Setting each $x_b=1/2$ satisfies the constraints for
$\lambda=1/2$. The objective is to round these $1/2$ values to
integral values in $\{0,1\}$ while respecting much weaker constraints,
which are given by $\lambda$-values as small as
$\lambda = 1/\poly\log n$. \Cref{thm:localsplit} can therefore
intuitively be interpreted as follows:

\begin{center}
\begin{minipage}{0.95\textwidth}
\begin{mdframed}[hidealllines=true, backgroundcolor=gray!20]
\emph{Coarsely rounding fractional numbers is essentially all that we do not know how to perform in $\poly\log n$ deterministic rounds of the \LOCAL model. If one can could do even coarse rounding in \polyloc, we could solve all the classic problems of the \LOCAL model in \polyloc.}
\end{mdframed}
\end{minipage}
\end{center}

As an intermediate step to prove the \polyseqloc-completeness of the local
splitting problems, we consider distributed algorithms for the
\emph{conflict-free multicoloring} problem. This is a natural relaxation of the \emph{conflict-free coloring} problem which was introduced in \cite{even03_conflictfree} in the context of frequency assignment in
cellular networks. Note that the relaxation only strengthens the completeness result.

\begin{definition}[\textbf{Conflict-Free
    Multicoloring}]\cite{even03_conflictfree}\label{def:conflictfree}
  \hspace*{5mm}A $q$-color multicoloring of a hypergraph \mbox{$H=(V,E)$} is a function
  $\phi: V\rightarrow 2^{[q]}\setminus \emptyset$ which assigns a
  nonempty subset $\phi(v)$ of the colors $[q]$ to each node $v$. A
  multicoloring $\phi$ is called \emph{conflict-free} if for each
  hyperedge $e\in E$, there exists at least one color $c$ such that
  $|\{v\in e| c\in \phi(v)\}|=1$, i.e., exactly one node in $e$ has
  color $c\in \phi(v)$.
\end{definition}

If each node is assigned exactly one color, such a coloring is called
a conflict-free coloring. Note that the conflict-free coloring
problem is a generalization of the standard graph coloring
problem. For a survey on various work related to conflict-free
coloring, we refer to \cite{smorodinsky2013conflict}.

\begin{theorem}\label{thm:conflictfree}
  Conflict-free multicoloring with $\poly\log n$ colors in almost
  uniform hypergraphs with $\poly n$ hyperedges is
  \polyseqloc-complete.
\end{theorem}

\subsection{Implications on Randomized Distributed Computations}

Because using randomization, an
$\big(O(\log n),O(\log n)\big)$-decomposition can be computed in
$O(\log^2 n)$ time in the \LOCAL model
\cite{awerbuch96,linial93,elkin16_decomp}, the \polyseqloc-completeness of
the decomposition problem (\Cref{thm:decomposition}) directly implies
that all problems in \polyseqloc\ have randomized polylog-time solutions in
the \LOCAL model. In fact, something slightly stronger holds. Let
$\randloc[t]{\eps}$ be the problems which can be solved by a
randomized Monte Carlo algorithm with error probability at most $\eps$
in the \LOCAL model in at most $t$ rounds. Further,
$\randseqloc[t]{\eps}$ is the corresponding randomized class for the
\SLOCAL model and we use $\polyrandloc{\eps}$ and
$\polyrandseqloc{\eps}$ to denote the corresponding randomized classes
of problems with polylogarithmic complexity.

\begin{theorem}\label{thm:randomizedequality} $\polyrandseqloc{\eps(n)} \subseteq \polyrandloc{\eps(n)+1/n^c}$
  for all $\eps(n)\geq 0$ and every constant $c>0$.
\end{theorem}

Hence, in particular, $\polyseqloc\subseteq \polyrandloc{1/\poly(n)}$.
In \Cref{sec:ILP}, we show that as long as all the constraints are
local, arbitrarily good approximations of general distributed covering
and packing integer linear programs can be computed efficiently in the
\SLOCAL model. This includes many important classic optimization
problems, such as e.g., minimum (weighted) dominating set, minimum
(weighted) vertex cover, maximum (weighted) independent set, maximum
(weighted) matching.

\begin{theorem}
  The problem of computing a $(1+1/\poly\log n)$-approximation of a
  general distributed covering or packing integer linear program (with
  polynomially bounded weights) is in \polyseqloc\ and hence also in
  $\polyrandloc{1/n^c}$ for every constant $c>0$.
\end{theorem}


\section{Computational Models and Complexity Classes}

\subsection{Distributed Graph Problems}
\begin{definition}[Distributed Graph Problem]\label{def:graphproblem}
  A \emph{distributed graph problem} $\calT$ is given by a set of triples of the form
  $(G,\vec{x},\vec{y})$, where $G=(V,E)$ is a simple, undirected graph
  and $\vec{x}$ and $\vec{y}$ are $|V|$-dimensional vectors with
  entries $x_v$ and $y_v$ for each node $v\in V$. We call $\vec{x}$
  the \emph{input vector} and $\vec{y}$ the \emph{output vector}. 
	 A tuple $(G,\vec{x})$ is called an $\emph{instance}$ of a graph problem $\calT$ if there is an output vector $\vec{y}$ such that $(G,\vec{x},\vec{y})\in \calT$. Then $\vec{y}$ is called an \emph{admissible} output for instance $(G,\vec{x})$.
	
Whether a triple belongs to $\calT$ or not depends only on the topology of
 the graph $G$. Hence, if there is an isomorphism mapping $G$ to $\tilde{G}$, then $(G,\vec{x},\vec{y})\in\calT$ holds if and only if  $(\tilde{G},\tilde{\vec{x}},\tilde{\vec{y}})\in\calT$ holds, where
  $\tilde{\vec{x}}$ and $\tilde{\vec{y}}$ are obtained from $\vec{x}$
  and $\vec{y}$ by applying the graph isomorphism
  from $G$ to $\tilde{G}$.
	\end{definition}
	Given an instance $\calI=(G,\vec{x})$ of a graph problem $\calT$, initially each node $v$ knows $x_v$. We always assume $x_v$ includes a unique ID for $v$ and a global polynomial upper bound on $n=|V|$. In a distributed algorithm, the nodes need to compute an admissible output vector $\vec{y}$, where each node $v\in V$ outputs $y_v$. For instance, consider $(\Delta+1)$-vertex coloring, where
$\Delta$ denotes the maximum degree of $G$. The problem
consists of all triples $(G,\vec{x},\vec{y})$, where $G=(V,E)$ is a
simple, undirected graph, \vec{x} contains  unique IDs, and
$y_v\in\set{1,\dots,\Delta+1}$ such that for each $\set{u,v}\in E$, we have  $y_u\neq y_v$.

\smallskip\noindent\textbf{Remark:}
For simplicity, we define inputs and outputs only for nodes. Edge related problems --- e.g., edge
coloring --- can be easily modeled as inputs and outputs to the incident nodes. Similarly, hypergraph problems can be modeled as graph
problems, where the locality is captured by a simple graph in which two nodes $u$ and $v$ are adjacent iff $u$ and $v$ are
in a common hyperedge.

\subsection{Distributed Local Algorithms}


In a distributed graph problem $\calT$ in the \LOCAL model, each node $v\in V$ of an instance $\calI=(G,\vec{x})$ initially learns its input $x_v$, and must output $y_v$ by the end of the algorithm. The \emph{time complexity}  of a $\LOCAL$ algorithm $\Alg$ on $\calI$ is the number of rounds until all nodes have completed the algorithm. Formally, the time complexity is a function $T_{\Alg}: \instances\rightarrow \N$, where $\instances$ is the set of all possible instances.

In the case of randomized \LOCAL\ algorithms, each node can produce an arbitrarily long private random bit string before it starts its computation. We focus on Monte Carlo randomized algorithms, which have fixed time complexity but may have some probability to err and produce an inadmissible output.
Let the random vector $\vec{y}$ denote the output vector of a randomized \LOCAL\ algorithm $\Alg$ on an instance $\calI=(G,\vec{x})$ of $\calT$. The error probability $\eps_{\Alg}(\calI)$ of $\Alg$ on $\calI$ is the probability that $(G,\vec{x}, \vec{y})\notin \calT$.

\subsection{Sequential Local Algorithms}
\label{sec:SLOCAL}

We define the sequential local model (\SLOCAL) as follows: Assume a problem instance \mbox{$\calI=(G,\vec{x})$} for $G=(V,E)$ is given. For each node $v\in V$, there is an unbounded local memory $S_v$ to
store the local state of $v$. Initially, $S_v$ contains only the private input $x_v$ of $v$. 
Then an algorithm $\Alg$ in the \SLOCAL model processes the nodes
sequentially in an order $p=v_1$, $v_2$, $\dots$, $v_n$ provided to $\Alg$. The algorithm must work for any given order $p$. When processing node
$v$, the algorithm can query
$r$-hop neighborhoods of node $v$ for different values of $r$, that is, $\Alg$ can read the values of
$S_u$ for all nodes $u$ in the $r$-neighborhood of $v$. 
Based on this information, node $v$ updates its state $S_{v}$ and
computes its output $y_{v}$. In doing so, node $v$ can perform unbounded computation, i.e., the new state of $S_{v}$ can be an arbitrary
function of the queried $r$-neighborhood of $v$. The output $y_v$ can be remembered as a part of the new value of
$S_{v}$. In randomized algorithms, each node $v$ produces an arbitrarily long private random bit string at the start of the execution (independent of $p$), which is stored in its initial state $S_v$.

The \emph{time complexity $T_{\Alg,p}(\calI)$} of the algorithm on $\calI$ with respect to order $p$ is defined as the maximum $r$ over all nodes $v$ for which the algorithm queries an $r$-hop neighborhood of node $v$. The algorithm's \emph{time complexity $T_{\Alg}(\calI)$} on instance $\calI$ is the maximum of all $T_{\Alg,p}(\calI)$ over all orders $p$.

Let the random vector $\vec{y}_p$ denote the output of a randomized \SLOCAL\ algorithm $\Alg$ on an instance $\calI=(G,\vec{x})$ of $\calT$ on node order $p$. The error probability $\eps_{\Alg}(\calI)$ of $\Alg$ on $\calI$ is $\max_{p} \Pr((G,\vec{x}, \vec{y}_p)\notin \calT)$.

\medskip\noindent\textbf{Remarks:} Many of the classic problems---e.g., maximal independent set, $(\Delta+1)$-vertex coloring, $(2\Delta-1)$-edge coloring, or maximal matching---can be solved in the \SLOCAL model with locality $O(1)$.  Roughly speaking, we can say that any problem in which any correct partial solution can be extended to a global solution using only local knowledge has a small locality in the \SLOCAL model. 



In studying \SLOCAL algorithms, it is convenient to allow nodes to write in the local memory of other nearby nodes. It is easy to see that this does not change the locality significantly. Concretely:
\begin{observation}\label{observation:memoryWriting}
Any \SLOCAL algorithm $\mathcal{A}$ with locality $R$ in which each node $v$ can write into the local memory $S_u$ of other nodes $u$ within its radius $r\leq R$ can be transferred into an $\SLOCAL$ algorithm $\mathcal{B}$ with locality $r+R$ in which $v$ writes only in its own memory $S_v$. 
\end{observation}
Furthermore, as explained above, the \SLOCAL model assumes a \emph{single-phase} of processing vertices in an order $p=v_1, v_2, \dots, v_n$. One can envision a generalization to \emph{$k$-phase} algorithms, which can go through the order $k$ times. However, perhaps somewhat surprisingly, for any $k\leq \poly\log n$, this generalization does not significantly increase the power of the model, as we prove in \Cref{lemma:phaseReduction}. Its proof, which is deferred to \Cref{sec:powerOfSequentialModel}, uses some techniques that are similar to those of \Cref{sec:networkDecomp}. 

\begin{lemma}
\label{lemma:phaseReduction}
\emph{Any  $k$-phase \SLOCAL\ algorithm $\mathcal{A}$ with locality $r_i$ in phase $i=1,\ldots,k$ can be transformed into a single-phase \SLOCAL\ algorithm $\mathcal{B}$ with locality $r_1+2\sum_{i=2}^kr_i$.}
\end{lemma}

\subsection{Complexity Classes}
\label{sec:classes}

We next define the complexity classes. Let $\instances$ be the collection of all instances $(G,\vec{x})$. 
A \emph{runtime function} is a function $t: \instances \rightarrow \ZZ^{+}$ and an \emph{error function } is a function $\eps: \instances \rightarrow [0,1]$. We say that an algorithm $\Alg$ has \emph{locality} $t$ if $T_{\Alg}(\calI)\leq t(\calI)$. We focus on (upper bound) runtime and error functions which depend only on the number of the graph vertices in $\calI$. Hence, we simply write $t(n)$ and $\eps(n)$. 

\begin{definition} For any runtime function $t$ and error function $\eps$ define:
\begin{description}
\item[{\boldmath\loc[t]:}] All graph problems $\calT$ for which there
  exists a \emph{deterministic distributed \LOCAL\ algorithm} $\calA$ such that for
  every instance $\calI$ of $\calT$, we have $T_{\Alg}(\calI)\leq t(\calI)$.
	
\item[{\boldmath\randloc[t]{\eps}}:] All graph problems $\calT$ for which there
  exists a \emph{randomized distributed \LOCAL\ algorithm} $\calA$ such that for
  every instance $\calI$ of $\calT$, we have $T_{\Alg}(\calI)\leq t(\calI)$ and $\eps_{\Alg}(\calI)\leq \eps(\calI)$.

\item[{\boldmath\seqloc[t]}:] All graph problems $\calT$ for which there
  exists a \emph{deterministic distributed \SLOCAL\ algorithm} $\calA$ such that for
  every instance $\calI$ of $\calT$, we have  $T_{\Alg}(\calI)\leq t(\calI)$.

\item[{\boldmath\randseqloc[t]{\eps}}:]All graph problems $\calT$ for which there
  exists a \emph{randomized distributed \SLOCAL\ algorithm} $\calA$ such that for
  every instance $\calI$ of $\calT$, we have  $T_{\Alg}(\calI)\leq t(\calI)$ and $\eps_{\Alg}(\calI)\leq \eps(\calI)$.
\end{description}
\end{definition}
Each deterministic class is trivially contained in its randomized counterpart. Moreover, the classes related to the $\LOCAL$ model are contained in their $\SLOCAL$ model counterparts. Concretely, if in a $\LOCAL$-algorithm, the nodes know an upper bound $r$ on the runtime, this can be transferred into an algorithm which first collects the $r$-hop neighborhood and then computes the output. Thus, we have:

\begin{lemma}
 $\loc[t(n)] \subseteq \seqloc[t(n)]$ and $\randloc[t(n)]{\eps} \subseteq \randseqloc[t(n)]{\eps}$, for every $t(n)$.
\end{lemma}

We use the $\bigO$-notation for runtime in the natural way: the class of all graph problems for which there is a sequential $O(t(n))$-local algorithm is denoted by \mbox{$\seqloc[\bigO(t(n))]:=\bigcup_{c>0}\seqloc[ct(n)]$}.
Our focus is on algorithms with polylogarithmic locality. We thus
introduce short notations for the above classes when the locality is
polylogarithmic in the number of nodes $n$:
\begin{align*}
  \boldsymbol\polyloc & :=  \bigcup_{c>0}\loc[\log^{c}n], 
	&  \text{\boldmath$\polyrandloc{\eps}$} & :=  \bigcup_{c>0}\randloc[\log^{c}n]{\eps}, \\
  \boldsymbol\polyseqloc & := \bigcup_{c>0}\seqloc[\log^{c}n], 
	&  \text{\boldmath$\polyrandseqloc{\eps}$} & :=  \bigcup_{c>0}\randseqloc[\log^{c}n]{\eps}.
\end{align*}

\subsection{Locality Preserving Reductions}
We now define reductions for distributed algorithms. An \emph{overlay graph} of a graph $G=(V,E)$ is a graph $\calG = (\calV, \calE)$, where each node $x\in \calV$ is mapped to a node $v(x)\in V$. An overlay graph $\calG$ is called $r$-simulatable if for every edge \mbox{$\set{x,y}\in \calE$}, we have \mbox{$d_G(v(x),v(y))\leq r$}. 
In our reductions, \LOCAL\ algorithms are augmented with oracles for a given graph problem $\calT$. After calling a $\calT$-oracle on overlay graph $\calG = (\calV, \calE)$, for each $x\in\calV$, node $v(x)\in V$ is provided with the output of the oracle for node $x$.

\begin{definition}[Reduction]
A \emph{(randomized) reduction} from a graph problem
$\calT_1$ to a graph problem $\calT_2$ is a (randomized) \LOCAL\ algorithm for $\calT_1$ which can use calls to a $\calT_2$-oracle with instances on overlay graphs of $G$. The \emph{cost of a reduction} is the cost of the \LOCAL\ algorithm where each oracle call on an $r$-simulatable overlay graph contributes $r$ rounds.
In the case of a randomized reduction, the randomness of all oracle instances and the reduction algorithm are independent. 
\end{definition}
As standard, a reduction from a graph problem $\calT_1$ to a graph problem $\calT_2$ transfers a $\LOCAL$ algorithm for $\calT_2$ to a $\LOCAL$ algorithm of $\calT_1$:

\begin{observation}
If there is a reduction from $\calT_1$ to $\calT_2$ and a $t_2(n')$ round \LOCAL\ algorithm for $\calT_2$ then there is a $t(n)\cdot t_2(n')$ round \LOCAL\ algorithm for $\calT_1$, where $t(n)$ is the cost of the reduction and $n'$ is the size of the largest overlay graph used in the reduction.
\end{observation}

\begin{definition}(Hardness and Completeness) We say that a graph problem $\calT$ is \emph{\complexClass-hard}, for a complexity class \complexClass\ , with respect to $t^{\bigO(1)}(n)$-cost reductions if every graph problem in \complexClass\ reduces to $\calT$ and the cost of each reduction is in $t^{\bigO(1)}(n)$. 
We say that $\calT$ is \emph{\complexClass-complete} with respect to $t^{\bigO(1)}(n)$-cost reductions if $\calT$ is \complexClass-hard with respect to $t^{\bigO(1)}(n)$-cost reductions and $\calT \in \complexClass$.
\end{definition}

Throughout the paper we are mostly interested in $\poly\log n$ cost reductions.  If a problem $\calT_1$ can be reduced to a problem $\calT_2$ with a polylog-cost reduction we say that $\calT_1$ is \emph{polylog-reducible} to $\calT_2$.


\section{Low Diameter Ordering \& Network Decomposition}
\label{sec:networkDecomp}
In this section we prove that the problems of computing a \emph{low diameter ordering} and \emph{network decomposition} are \polyseqloc-hard. See \Cref{def:ordering} and \Cref{def:decomposition} for the definitions, respectively.\\

\noindent \textbf{Notation:} For a graph $G=(V, E)$, we use $G^{r}$ to denote the graph on vertex set $V$ obtained by putting an edge between each two vertices of $G$ with distance at most $r$. 

\begin{lemma}\label{lemma:orderingHard}
  For any $d(n)=\log^{O(1)}n$, computing a $d(n)$-diameter ordering is \polyseqloc-hard.
\end{lemma}
\begin{proof}
  Consider an \SLOCAL\ algorithm $\mathcal{A}$ with locality $r=\poly\log n$ and consider an $n$-node network graph $G=(V, E)$. Assume that an $\ell$-low diameter ordering $\pi$ of the graph $G^{r}$ is provided by an oracle where $\ell=\poly\log n$. We consider $\mathcal{A}$ when it operates on the order $\pi$.

When processing the nodes according to the order $\pi$, a node $v$ can collect its $r$-neighborhood and compute its output, as soon as all nodes within distance $r$ of $v$ which appear before $v$ in order $\pi$ are processed. Hence, every path on $G^r$ which is monotonically increasing w.r.t.\ $\pi$ induces a dependency chain for executing $\calA$. Given that $\pi$ is an $\ell$-diameter ordering of $G^r$, each such dependency chain, which is relevant for processing $v$, is completely contained in the $\ell$ neighborhood of $v$ in $G^r$. After collecting the $r$-neighborhoods in $G$ of every node in the $\ell$-neighborhood in $G^r$, node $v$ therefore has enough information to locally simulate the part of the sequential execution of $\calA$ which is relevant for processing node $v$. Thus, given an $\ell$-diameter ordering of $G^r$, algorithm\ $\calA$ can be executed in $O(\ell r)$ deterministic rounds in the \LOCAL model.
 \end{proof}

The best known deterministic algorithms for many problems in \polyseqloc\ are based on network decompositions (cf.\ \Cref{def:decomposition}). In fact network decompositions directly imply low-diameter orderings and thus are sufficient to simulate polylogarithmic \SLOCAL algorithms.

\begin{observation}\label{observation:lowDiameterNetworkDecomp}
If we are given a $\big(d(n), c(n)\big)$-decomposition and
assign to each vertex $v\in G$  a label $(q_v, \mathsf{ID}_v)$ where $q_v$ is the color of $v$'s cluster, then the lexicographically increasing order of the node labels $(q_v, ID_v)$ defines a $O\big(d(n)\cdot c(n)\big)$-diameter ordering.
\end{observation}

\begin{lemma}
\label{lemma:networkdecompositionhard} Computing a
$(\poly\log n, \poly\log n)$-decompositon is 
$\polyseqloc$-hard.
\end{lemma}
\begin{proof}
The result follows with \Cref{lemma:orderingHard} and \Cref{observation:lowDiameterNetworkDecomp}.
\end{proof}

The completeness of low diameter orderings and network decompositions (\Cref{thm:decomposition} and \Cref{thm:ordering}) follows by an adaption of the deterministic sequential $\big(\bigO(\log n),\bigO(\log n)\big)$-decomposition algorithm from \cite{linial93} to the \SLOCAL model (cf.\ \Cref{appsec:inPSLOCAL}).

The network decomposition algorithm of
Awerbuch et al.~\cite{awerbuch89} computes a $\big(2^{\bigO(\sqrt{\log n})}, 2^{\bigO(\sqrt{\log n})}\big)$-decomposition deterministically in the \LOCAL model. This algorithm combined with \Cref{observation:lowDiameterNetworkDecomp} and the same simulation as in the proof of \Cref{lemma:orderingHard} yields the following lemma.
\begin{lemma}
$\seqloc\big(2^{\bigO(\sqrt{\log n})}\big) = \loc\big(2^{\bigO(\sqrt{\log n})}\big)$.
\end{lemma}

\begin{remark}
In general, for $t(n)\geq \log n$,
$\seqloc\big(t^{\bigO(1)}(n)\big) = \loc\big(t^{\bigO(1)}(n)\big)$ holds if and only if a
$\big(t^c(n),t^c(n)\big)$-network decomposition can be computed
deterministically in $O\big(t^c(n)\big)$ rounds in the \LOCAL model for some
constant $c>0$.  
\end{remark}



\section{Overview of Local Splitting Completeness Proof}
\label{sec:LocalSplitSummary}

In the present section, we provide an outline over the proof that the
\emph{local splitting} problems defined in
\Cref{def:split,def:weaksplit} are \polyseqloc-complete. The formal
proof appears in \Cref{sec:conflictfree,sec:localsplit}. We need to
show that local splitting is in the class \polyseqloc\ and that local
splitting is \polyseqloc-hard, i.e., that there is a
polylog-reduction, reducing one of the problems we have already shown
to be \polyseqloc-complete to local splitting. We do this reduction in
two steps. We first reduce the conflict-free multicoloring problem (cf.\ \Cref{def:conflictfree}) to local splitting
and we then reduce the problem of computing a
$(\poly\log n,\poly\log n)$-decomposition to the conflict-free
multicoloring problem (cf.\ \Cref{def:decomposition} and
\Cref{thm:decomposition}).

\subsection{Reducing Conflict-Free Multicoloring to Local Splitting}

We next sketch how to use a $\lambda$-local splitting blackbox
algorithm (for $\lambda=1/\poly\log n$) to compute a $\poly\log n$-color conflict-free multicoloring
of a given $n$-node hypergraph $H=(V,E)$. A reduction to weak local
splitting then follows by applying a simple reduction from
$\lambda$-local splitting which we describe in
\Cref{lemma:weaksplithardness}.

By using a distributed \emph{defective coloring} algorithm from
\cite{Kuhn2009WeakColoring}, we first show in
\Cref{lemma:lowrank_cfcoloring} that for hypergraphs of at most
$\poly\log n$ rank, a $\poly\log n$-color conflict-free multicoloring
can be computed in deterministic polylog time in the \LOCAL model. The
reduction then works in phases, where in each phase, we remove some
hyperedges and nodes from $H$. We define $\delta:=1/\lambda$, note
that this implies that $\delta=\poly\log n$. In each phase, we first
apply \Cref{lemma:lowrank_cfcoloring} and assign a new set of
$\poly\log n$ colors to make sure that for all hyperedges $e$ of rank
at most $\delta$ of the current graph $H$, there exists a color $x$
such that exactly one node in $e$ has color $x$. This allows to remove
all hyperedges of rank at most $\delta$. We then interpret the
resulting hypergraph $H$ as a bipartite graph in the obvious way and
we apply our $\lambda$-local splitting oracle to this bipartite graph
so that all nodes of $H$ are either colored red or blue and so that
each hyperedge $e$ has at least $\lfloor\lambda|e|\rfloor$ nodes of
each color. We then remove all blue nodes from the graph $H$. Because
after the removal of the low-rank hyperedges, all hyperedges have rank
$>\delta$, the $\lambda$-local splitting guarantees that each
hyperedge has at least one red node. Therefore a conflict-free
multi-coloring of the hypergraph after removing the blue nodes is also
a conflict-free multi-coloring of the hypergraph before removing the
blue nodes. Because each hyperedge $e$ has at least
$\lfloor\lambda|e|\rfloor$ blue nodes which are removed, the removal
of the blue nodes reduces the maximum rank of the hypergraph $H$ by a
factor $1-1/\Theta(\lambda)$. Because the maximum rank at the
beginning is at most $n$, the number of phases is
at most $O(\log(n)/\lambda)$ and thus $O(\poly\log n)$. Thus, the
number of colors that we use for the conflict-free multicoloring is
also $O(\poly\log n)$.

\subsection{Reducing Network Decomposition to Conflict-Free
  Multicoloring}

We conclude this section by giving an overview of how to use
conflict-free multicoloring to compute a network decomposition. The
resulting decomposition algorithm bears some high-level similarities
to existing randomized graph decomposition algorithms (e.g.,
\cite{linial93,blelloch14,elkin16_decomp}). Assume that we have a
$q$-color conflict-free multicoloring algorithm for almost uniform
$n$-node hypergraphs for some $q=\poly\log n$ and assume that we need
to compute a $(\poly\log n,\poly\log n)$-decomposition of some graph
$G=(V,E)$. As a first step, each node $v\in V$ looks for a sequence of
$q+1$ consecutive radii $r_v,r_v+1,\dots,r_v+q$ such that all the
balls $B_{r_v+i}(v)$ for $i\in\set{0,\dots,q}$ have the same size up
to a factor $(1\pm \eps)$ for a given constant $\eps>0$. Using
standard ball growing arguments
\cite{awerbuch85,Awerbuch-Peleg1990,linial93}, there exists such a
radius of value $r_v=O(q\log(n)/\eps)$ and clearly in the \LOCAL model
such a radius can then also be found in $O(q\log(n)/\eps)$ rounds for
each node.

Each node now forms $q+1$ hyperedges for its balls
$B_{r_v}(v),\dots,B_{r_v+q}(v)$ and the reduction constructs
$O(\log(n)/\eps)$ hypergraphs such that in each of them all
hyperedges have the same size up to a $(1\pm O(\eps))$-factor and such
that all hyperedges of a given node $v$ are in the same
hypergraph. For each of these hypergraphs, we use the conflict-free
multicoloring oracle to compute a $q$-color conflict-free
multicoloring. Because each node has $q+1$ hyperedges and the nodes in
these hyperedges are conflict-free colored with $q$ colors, by the
pigeonhole principle, there is a color $x$ and two radii $r_v+a$ and
$r_v+b$ for $0\leq a< b\leq q$ such that in $B_{r_v+a}(v)$ and
$B_{r_v+b}(v)$, there is exactly one node $w$ colored with color
$x$. Clearly, for both balls, it has to be the same node. Node $v$
chooses this node $w$ as its ``cluster center'' and it chooses color
$x$ as its cluster color. Because node $w$ is within radius $r_v+a$ of
$v$ and there is no other node of color $x$ within radius
$r_v+b\leq r_v+a+1$ of $v$, whenever a neighbor $u$ of $v$ also
chooses color $x$, node $u$ also has to choose $w$ as its cluster
center. Hence, for every cluster color, any two nodes within the same
connected component have the same cluster center and are thus within
radius $O(q\log(n)/\eps)$ in graph $G$. As we assumed that
$q=\poly\log n$, this implies that the computed coloring directly
induces a $(\poly\log n,\poly\log n)$-decomposition.

\subsection{Weak Local Splitting is in \boldmath\polyseqloc}

We here only discuss how to design an \SLOCAL\ algorithm with polylog
locality to compute a weak local splitting for a given bipartite graph
$B=(U\dot{\cup}V,E_B)$ where each node in $U$ has degree $c\ln^2 n$
for a sufficiently large constant $c$. An algorithm for
$\lambda$-local splitting can then be obtained by using a simple
reduction, which is described in \Cref{lemma:weaksplithardness}.
Using \Cref{lemma:phaseReduction}, we can design a multi-phase \SLOCAL
algorithm to show that weak local splitting is in \polyseqloc. The
algorithm is based on first computing a
$(O(\log n),O(\log n))$-decomposition of the graph $G=(V,E)$, where
there is an edge between $u$ and $v$ in $V$ if and only if $u$ and $v$
have a common $U$-neighbor in $B$. It is shown in
\Cref{lemma:networkDecomp} that computing such a decomposition is in
\polyseqloc. We can use the network decomposition to compute weak
local splitting as follows. Each cluster locally computes a
red/blue-coloring of its nodes in $O(\log n)$ rounds. The
probabilistic method guarantees that each cluster $\calC$ can compute
such a coloring such that for every node $u$ of which $\calC$ contains
at least $d\ln n$ neighbors for a sufficiently large constant $d$, the
neighborhood $N(u)$ becomes bichromatic. The decomposition guarantees
that the neighborhood of each node $u\in U$ is partitioned among at
most $O(\log n)$ clusters. Because we assume that the minimum degree
in $U$ is at least $c\ln^2 n$ (for $c$ sufficiently large), for every
node $u\in U$, there is a cluster which contains at least $d\ln n$
neighbors of $u$. Hence, we get a weak local splitting of the whole
graph.


\section{Completeness of Conflict-Free Multicoloring}
\label{sec:conflictfree}

In the present section, we study the distributed complexity of conflict-free
multicoloring of hypergraphs (cf.\
\cite{even03_conflictfree,smorodinsky2013conflict} and
\Cref{def:conflictfree}). Recall that a $q$-color conflict-free
multicoloring of a hypergraph $H=(V,E)$ is an assignment of a nonempty
set $\phi(v)$ of colors from $[q]$ to each node $v\in V$ such that for
every hyperedge $e\in E$, there is a color $x$ such that there is
exactly one node in $e$ which has color $x$ in its set $\phi(v)$. Note
that in the special case of simple graphs, when each hypergraph
contains only a pair of nodes, and if only one color is allowed per
node, conflict-free coloring is equivalent to the standard definition
of proper graph coloring.

In the following, for a given constant $0<\eps<1$, we call a
hypergraph $H=(V,E)$ \emph{almost uniform} if there is an arbitrary
$k$ such that for each edge $e\in E$,
\mbox{$k\leq |e|\leq (1+\eps)k$}. In the following, we prove
\Cref{thm:conflictfree}.

\medskip 
\noindent\textbf{\Cref{thm:conflictfree} (restated).}
\emph{Conflict-free multicoloring with $\poly\log n$ colors in almost
  uniform hypergraphs with $\poly n$ hyperedges is
  \polyseqloc-complete. } 
\smallskip
\begin{proof}
  The proof follows directly from the statements of
  \Cref{lemma:conflictfreeHardness,lemma:CFcontained} which are proven
  next in \Cref{sec:CFhard,sec:CFcontained}.
\end{proof}

Before presenting the proofs of
\Cref{lemma:conflictfreeHardness,lemma:CFcontained}, we remark that
conflict-free multicoloring is trivial to solve using randomized
algorithms with even zero locality.

\begin{observation} There is a zero round randomized \LOCAL algorithm that in any almost uniform hypergraph with $\poly(n)$ hyperedges computes an $O(\log n)$-color conflict-free \multicoloring . 
\end{observation}
\begin{proof}
  Set $q=\Theta(\log n)$ and define a \multicoloring\
  $\phi:V\rightarrow 2^{[q]}\setminus \emptyset$ by including each
  color $c\in [q-1]$ in $\phi(v)$ with probability $\frac{1}{k}$. If
  $\phi(v)=\emptyset$, set $\phi(v)=\{q\}$. This is a conflict-free
  coloring, with high probability: for each hyperedge $e\in E$ and
  each color $c\in [q-1]$,
  $\Pr[|\{v\in e| c\in \phi(v)\}|=1] \geq
  |e|\frac{1}{k}\left(1-\frac{1}{k}\right)^{|e|-1} \geq 0.1$.
  Hence, the probability that no color $c\in [q-1]$ satisfies
  $|\{v\in e| c\in \phi(v)\}|=1$ is $0.9^{C\log n} = 1/\poly(n)$. A
  union bound over all hyperedges $e\in E$ completes the proof.
\end{proof}

\subsection{Conflict-Free Multicoloring is $\polyseqloc$-Hard}
\label{sec:CFhard}

We now first show that conflict-free multicoloring of almost uniform
hypergraphs with $\polylog n$ colors is \polyseqloc-hard. We show this
by showing that the problem of computing a
$(\polylog n, \polylog n)$-decomposition of a graph is
polylog-reducible to the conflict-free multicoloring problem.

\begin{lemma} \label{lemma:conflictfreeHardness}
  The problem of computing a $q$-color conflict-free multicoloring of
  an almost uniform hypergraph with $q=\poly\log n$ colors is
  \polyseqloc-hard.
\end{lemma}

\begin{proof}
  Assume that for some given $0<\eps<1$ and $q=\log^{O(1)}n$, such
  that for every $k\leq n$, we have an oracle to compute
  a $q$-color conflict-free multicoloring of a given $n$-node
  hypergraph $H=(V_H,E_H)$ with polynomially many hyperedges and where
  for each hyperedge $e\in E_H$,
  $k\leq |e|\leq (1+\eps/3)^2k < (1+\eps)k$. We use
  $O(\log n/\eps)$ iterations of the $q$-color multicoloring oracle to compute a
  $(\poly\log n, \poly\log n)$-decomposition of a given graph
  $G=(V,E)$ in polylogarithmic deterministic time in the \LOCAL model.
  Since \Cref{lemma:networkdecompositionhard} shows that
  $(\poly\log n, \poly\log n)$-network decomposition is
  \polyseqloc-hard, we get that $q$-color multicoloring of almost
  uniform hypergraphs is also \polyseqloc-hard.

  \paragraph{Construction of the hypergraphs {\boldmath$H_1$, \dots, $H_\ell$}.}  %
	We first define $\ell=O(\log n/\eps)$ almost uniform hypergraphs
  \mbox{$H_1$, $H_2$, \dots, $H_\ell$} on the node set $V$.
	For each vertex $v$, let $B_r(v)$ denote the set of all vertices within distance $r$ of $v$ in graph $G$. Let $r_v$ be the smallest radius $r$ such that $\frac{|B_{r+q}(v)|}{|B_r(v)|}\leq 1+\eps/3$. Note that $r_v\leq \bigO(q\log n /\eps)$. This is because, otherwise, with every $q$ additive increase in the radius of the ball $B_r(v)$, its size would grow by a $(1+\eps/3)$ factor and this cannot happen more than $\bigO(\log n/\eps)$ many times. 
	Include $q+1$ hyperedges, each defined by one of the vertex sets $B_{r_v} (v)$, $B_{r_v+1} (v)$, \dots, $B_{r_v+q} (v)$, all in the hypergraph $H_i$ such that $i=\lfloor \log_{1=\eps/3} |B_{r_v}(v)|\rfloor$. Perform this for each vertex $v$.  
	
	Note that every
  node $v$ can perform this step and define its hyperedges in the
  \LOCAL model in $O\big(q\frac{\log n}{\eps}\big)$
  rounds.
	  Notice that each hypergraph $H_i$ is almost uniform because each hyperedge
$e\in H_i$ has size $(1+\eps/3)^i \leq |e| < (1+\eps/3)^2 \cdot
  (1+\eps/3)^{i}.$
  Furthermore, each hyperedge of each $H_i$ has radius at most
  $R=\bigO(q\log n/\eps)$ in $G$ and thus a round of communication on
  $H_i$ can be simulated in $\bigO(q\log n/\eps)$ rounds on $G$.

  \paragraph{Construction of network decomposition.} We make $\ell = \bigO(\log n/\eps)$ (parallel) calls to the
  $q$-color multicoloring oracle to compute a $q$-color conflict-free
  multicoloring for each hypergraph $H_i$, where the coloring of $H_i$
  for $i\in [\ell]$ uses colors in $[1+(i-1)q, iq]$. We claim that
  this provides a $(2R, q\ell)$-decomposition.

  Define the network decomposition as follows. For each vertex
  $v\in G$, we define a cluster center and a cluster color.  The
  cluster centers and cluster colors are 
  defined as follows. Consider the hyperedges corresponding to the hyperedges
  $B_{r_v} (v)$, $B_{r_v+1} (v)$, \dots, $B_{r_v+q} (v)$ in $H_i$.
  Associate each of these hyperedges $e$ with one color
  $c\in [1+(i-1)q, iq]$ such that exactly one vertex in $e$ has color
  $c$. Note that such a color exists by the definition of a conflict-free
  multicoloring. Since there are $q+1$ hyperedges, one corresponding
  to each ball, and they are associated with only $q$ colors in
  $[1+(i-1)q, iq]$, by the pigeonhole principle, there are two radii
  $r_1, r_2 \in [r_{v}, r_v+q]$, $r_1<r_2$, such that the hyperedges
  corresponding to $B_{r_1}(v)$ and $B_{r_2}(v)$ are associated with
  the same color $c\in [1+(i-1)q, iq]$. Therefore, there is a node
  $u \in B_{r_1}(v)$ that is colored with color $c$ and this is the
  only vertex in $B_{r_2}(v)$, and thus also in $B_{r_1+1}(v)$, that is
  colored with color $c$. Then, $v$ will be in a cluster of color $c$
  and the cluster-center $\mathit{Center}(v):=u$. Notice that when defining
  $t:=dist(v, \mathit{Center}(v))$, we have the following \emph{uniqueness}
  property: the node $\mathit{Center}(v)=u$ is the only node within distance $t+1$
  of $v$ that has color $c$.

  To prove that we get a $(2R, q\ell)$-network decomposition, we argue
  that for each two neighboring nodes $v_1$ and $v_2$ which are in clusters
  of the same color $c$, we have $\mathit{Center}(v_1)=\mathit{Center}(v_2)$. Let
  $t_1=dist(v_1, \mathit{Center}(v_1))$ and $t_2=dist(v_2, \mathit{Center}(v_2))$.
  Suppose that $t_1\geq t_2$. Then, $\mathit{Center}(t_2)$ is within distance
  $t_1+1$ of $v_1$. By the uniqueness property stated above,
  $\mathit{Center}(v_1)$ is the only node within distance $t_1+1$ of $v_1$ that
  has color $c$. Hence, $\mathit{Center}(v_1)=\mathit{Center}(v_2)$.
  We therefore get that any connected component of the same color has
  weak diameter at most $2R$, which concludes the proof.
\end{proof}

\subsection{Conflict-Free Multicoloring is in $\polyseqloc$}
\label{sec:CFcontained}

\begin{lemma} \label{lemma:CFcontained} 
  There is an \SLOCAL algorithm with locality $\poly\log(n)$ that
  finds an $\bigO(\log n)$-color conflict-free
  \multicoloring\footnote{For sufficiently large $k$,
    this algorithm can be modified to a coloring which assigns each
    node exactly one color.} in any almost uniform hypergraph with
  $\poly(n)$ hyperedges.
\end{lemma}

\begin{proof}
  We describe an $\bigO(\log n)$-phase conflict-free \multicoloring\
  algorithm. This can be transferred into a single-phase algorithm
  using \Cref{lemma:phaseReduction}. Using
  \Cref{observation:memoryWriting}, we can also assume that when
  processing a node $u$, $u$ can write into the memory of nodes in its
  polylog-neighborhood. In each phase, we use one new color such that
  the number of hyperedges which do not have a unique color reduces by
  a constant factor. Consider the $i^{th}$ phase. Suppose
  that $v_1$, $v_2$, \dots, $v_n$ is the provided order and we are now
  working on $v_j$. Let $B_{r}(v)$ be the set of vertices within
  distance $r$ of $v$ and $E[B_{r}(v)]$ be the set of hyperedges with
  all their vertices in $B_{r}(v)$. We check the $\bigO(\log n)$
  neighborhood of $v_j$ to see if $v_j$ is \emph{processed} before in
  the $i^{th}$ phase.  Otherwise, we use a ball growing method to find
  a radius $r\leq R=\bigO(\log n)$ such that
  $\frac{|E[B_{r+2}(v_j)]|}{|E[B_{r}(v_j)]|} \leq 2$.  Then, $v_j$
  assigns color $i$ to some of the vertices in $B_{r}(v_j)$ such that
  a constant fraction of the hyperedges in $E[B_{r}(v)]$ have exactly
  one vertex with color $i$. Such a coloring exists, by a
  probabilistic method argument: coloring each $u\in B_{r}(v_j)$ with
  color $i$ with probability $\frac{1}{k}$ would provide such a
  coloring, with a positive probability. Then, all nodes in
  $B_{r+1}(v)$ are considered \emph{processed} for phase $i$; they
  will not be colored again in this phase.

Since $|E[B_{r}(v)]| \geq |E[B_{r+2}(v)]|/2$, this process removes a constant fraction of the edges incident on the newly processed nodes. Hence, at the end of the phase, at least a constant fraction of the hyperedges of this phase have received unique colors, i.e., having exactly one vertex with color $i$. Since per phase a constant fraction of the remaining hyperedges receive unique colors, $\bigO(\log n)$ phases suffice. At the end, vertices with no color are assigned a default color.
\end{proof}


\section{Completeness of Local Splitting}
\label{sec:localsplit}

In this section, we discuss the \emph{local splitting} problems
defined in \Cref{def:split,def:weaksplit} and we show that these
extremely rudimentary looking problems in some sense capture the core
of the difficulty in designing $\poly\log n$ round deterministic
algorithms in the \LOCAL model. As outlined in
\Cref{sec:LocalSplitSummary}, we reduce the conflict-free
multicoloring problem to the local splitting problems. For
completeness, we restate the definitions of $\lambda$-local splitting
and weak local splitting.

\medskip

\noindent\textbf{\Cref{def:split}} (\textbf{Local Splitting})\textbf{.}
Given is a bipartite graph $B=(U \dot{\cup} V, E_B)$ where
$E_B\subseteq U\times V$. For any $\lambda\in [0,1/2]$, we define a
$\lambda$-local splitting of $B$ to be a $2$-coloring of the nodes in
$V$ with colors red and blue such that each node $v$ has at least
$\lfloor\lambda\cdot d(v)\rfloor$ neighbors of each color.

\medskip

\noindent\textbf{\Cref{def:weaksplit}} (\textbf{Weak Local Splitting})\textbf{.}
  Given is a bipartite graph $B=(U \dot{\cup} V, E_B)$ where \mbox{$E_B\subseteq
  U\times V$}. We define a weak local
  splitting of $B$ to be a $2$-coloring of the nodes in $V$ with
  colors red and blue such that each node $v$ has at $1$ neighbor of
  each color.

\medskip

Note that if $\delta$ is the minimum degree of any node $v\in U$ and
we set $\lambda\geq 1/\delta$, any $\lambda$-local splitting is also a
weak local splitting. We will show that even for graphs where all
nodes $v\in U$ have degree $\delta/2<d(v)\leq \delta$, there is a
constant $c>0$ such that for any $\delta$ with
$c\ln n \leq \delta = \log^{O(1)} n$, the weak local splitting problem
is \polyseqloc-complete. We prove this by first reducing the problem
of computing a $1/\poly\log n$-local splitting to the weak local
splitting problem and by then reducing the conflict-free
multicoloring of the previous section to the problem of computing a
$1/\poly\log n$-local splitting of a given bipartite graph $B$. We
note that the weak local splitting problem can be seen as a
generalization of the weak $2$-coloring problem introduced and studied
in \cite{naor95}. A weak $2$-coloring of a graph $G$ is a $2$-coloring
of the nodes of $G$ such that each node has at least one neighbor of a
different color. If we define a hypergraph $H$ with the same set of
nodes as $G$ and where we add a hyperedge for each of the $n$
$1$-neighborhoods of $G$, a weak local splitting of the bipartite
graph corresponding to $H$ is exactly a weak $2$-coloring of $G$.
Using techniques from \cite{naor95,linial92,Kuhn2009WeakColoring}, the
weak $2$-coloring problem can be solved in $O(\log^*n)$ deterministic
rounds in the \LOCAL model. Hence, the weak local splitting problem
can be solved efficiently for some interesting special cases. We
however show that even in sparse bipartite graphs, the general case is
as hard as any \polyseqloc-problem.

Before proving the hardness of the local splitting problems, we point
out that both local splitting problems are trivially solvable without
communication when using randomization. 

\begin{observation}\label{obs:split}
  There are positive constants $c$ and $\eps<1/2\sqrt{c}$ such that
  there is a zero-round randomized distributed algorithm which,
  w.h.p., solves the $\big(1/2-\eps\sqrt{\ln(n)/\delta}\big)$-local splitting
  problem for every bipartite $n$-graph $B=(U\dot{\cup}V, E)$ in which
  each node in $U$ has degree at least $\delta\geq c\ln n$.
\end{observation}

The randomized algorithm is trivial: Each node in $V$ is independently
colored red or blue with probability $1/2$. \Cref{obs:split}
then directly follows from Chernoff bounds and a union bound over all
nodes in $U$.

Specifically, the goal of this section is to prove
\Cref{thm:localsplit,thm:weaklocalsplit}, which we restate here for
completeness. The theorems directly follow from the technical lemmas
which appear in the next two subsections.

\medskip 
\noindent\textbf{\Cref{thm:localsplit} (restated).}
\emph{For $n$-node bipartite graphs $H=(U\dot{\cup} V, E)$ where all
  nodes in $U$ have degree at least $c\ln^2 n$ for a large enough
  constant $c$, the $\lambda$-local splitting problem for any
  $\lambda=\frac{1}{\poly\log n}$ is \polyseqloc-complete.}
\smallskip
\begin{proof}
  The claim directly follows from \Cref{lemma:localsplithardness,lemma:seqlocalsplit}.
\end{proof}

\medskip 
\noindent\textbf{\Cref{thm:weaklocalsplit} (restated).}
\emph{For $n$-node bipartite graphs $H=(U\dot{\cup} V, E)$ where all nodes in $U$
  have degree $\delta/2< d(u)\leq \delta$, for any $\delta$ such that
  $c\ln^2 n \leq \delta = \log^{O(1)} n$ for a sufficiently large
  constant $c$, the weak local splitting problem is
  \polyseqloc-complete.}
\smallskip
\begin{proof}
  The claim directly follows from
  \Cref{lemma:weaksplithardness,lemma:localsplithardness,lemma:seqlocalsplit}.
\end{proof}

\subsection{Weak Local Splitting is \boldmath\polyseqloc-hard}

As a part of the reduction for the local splitting problem, we need to
show that in hypergraphs of polylogarithmic rank, the conflict-free
multicoloring problem is in \polyloc. This is proven by the following
lemma.

\begin{lemma}\label{lemma:lowrank_cfcoloring}
  Let $H=(V,E)$ be an $n$-node hypergraph of rank at most
  $\kappa=\log^{O(1)}n$. There exists a
  $q=\log^{O(1)}n$ such that a $q$-color conflict-free
  multicoloring of $H$ can be computed in deterministic $\poly\log n$
  time in the \LOCAL model.
\end{lemma}
\begin{proof}
  The solution consists of $\ell$ phases, where in each phase, we
  remove some of the hyperedges from $H$. Let $H_i$ be the hypergraph
  before starting phase $i$, i.e., we have $H_1=H$. In each phase $i$,
  we color the nodes with colors from a new set $C_i$ of
  $c=O(\kappa^2\log n)$ colors and we afterwards remove all hyperedges
  from $H_i$ which contain exactly one node with color $x$ for some
  $x\in C_i$. The process ends when all hyperedges are removed. We
  show that this can be done such that the number of phases $\ell$ is
  polylogarithmic in $n$, implying the statement of the lemma.

  Let us now have a closer look at a specific phase $i$. We define
  a multigraph $G_i$ based on $H_i$. $G_i$ has the same node set as
  $H_i$ and the following edge set: We add one edge between every two
  nodes $u, v \in G_i$ for every hyperedge $e\in H_i$ that includes
  both $u$ and $v$. Hence, if the two nodes $u$ and $v$ share $\ell$ hyperedges,
  we include $\ell$ parallel edges between $u$ and $v$ in $G_i$. Let
  $\Delta_i$ be the maximum degree of $G_i$ (where the degree of a
  node is the number of its edges).

  In order to color the vertices of the hypergraph $H_i$ in phase $i$, we apply a
  distributed \emph{defective coloring} algorithm of
  \cite{Kuhn2009WeakColoring} to $G_i$. Given a graph $G$, a
  $d$-defective $c$-coloring of $G$ is a $c$-coloring of the nodes of
  $G$ such that every node has at most $d$ neighbors of the same
  color. In \cite{Kuhn2009WeakColoring}, it is shown that in an $n$-node graph
  with maximum degree $\Delta$, for any $p\geq 1$, one can compute a
  $(\Delta/p)$-defective $O(p^2\log n)$-coloring in a single
  communication round. From the construction in
  \cite{Kuhn2009WeakColoring}, it is straightforward to see that the
  algorithm can also directly be applied to multigraphs, where in a
  $d$-defective coloring of a multigraph, each node must be in at most
  $d$ monochromatic edges. 

  Using the algorithm from
  \cite{Kuhn2009WeakColoring}, we compute a
  $(\Delta_i/2\kappa)$-defective $O(\kappa^2\log n)$-coloring of the
  multigraph $G_i$. Now, each node of $H_i$ has one of
  $c=O(\kappa^2\log n)$ colors. As stated, to obtain the hypergraph
  $H_{i+1}$ for phase $i+1$, we now remove every
  hyperedge $e$ from $H_i$ for which there exists a color $x$ among these
  $O(\kappa^2\log n)$ colors such that exactly one node in $e$ has
  color $x$. Let $G_{i+1}$ be the multigraph which we obtain from the
  resulting hypergraph $H_{i+1}$ and as before, let $\Delta_{i+1}$ be
  the maximum degree of this multigraph. We next show that
  $\Delta_{i+1}\leq\Delta_i/2$. The claim of the lemma then directly
  follows because initially, each node can be in at most
  ${n-1\choose k-1}\leq n^{\kappa-1}$ hyperedges and thus the the
  maximum degree $\Delta_1$ of $G_1$ is at most quasi-polynomial in
  $n$. Therefore, the number of phases is at most
  $O(\log \Delta_1)=\log^{O(1)}n$.

  It remains to show that $\Delta_{i+1}\leq \Delta_i/2$. Consider a node
  $u$ and its incident hyperedges in $H_i$. Notice that a hyperedge
  $e\in H_i$ of node $u$ will remain for $H_{i+1}$ only if at least one other node
  $v\in e$ receives the same color as the color assigned to node
  $u$. In this case, the corresponding edge $\set{u,v}$ in $G_i$ is
  monochromatic. Since the coloring of $G_i$ has defect at most
  $\frac{\Delta_i}{2\kappa}$, we know that in $G_i$ there are at most
  $\frac{\Delta_i}{2\kappa}$ monochromatic edges incident to $u$. Hence,
  it follows $H_{i+1}$ can have at most $\frac{\Delta_i}{2\kappa}$
  hyperedges that contain node $u$. Given that each hyperedge of
  $H_{i+1}$, which is also a hyperedge of the original hypergraph $H$,
  has at most $\kappa$ nodes, we get that in $G_{i+1}$, node $v$ has
  degree at most $\Delta_i/2$. Thus, $\Delta_{i+1}\leq \Delta_i/2$. 
\end{proof}

We next show that there is a simple reduction from the $1/\poly\log
n$-local splitting problem to the weak local splitting problem.

\begin{lemma}\label{lemma:weaksplithardness}
  Let $\delta$, for $2\leq \delta=\log^{O(1)}n$, be an integer parameter and let
  $\lambda=1/\delta$. The $\lambda$-local splitting
  problem in $n$-node bipartite graphs $B=(U \dot{\cup} V, E)$ is
  polylog-reducible to the weak splitting problems in a bipartite
  graph $B'=(U'\dot{\cup}V, E')$ where each node $v$ in $U'$ has degree
  $\delta/2<d(v)\leq\delta$.
\end{lemma}
\begin{proof}
  By using an oracle for the weak splitting problem, we need to
  deterministically solve the $\lambda$-local splitting problem on $B$
  in polylog time in the \LOCAL model. Note that we can w.l.o.g.\
  assume that all nodes $u\in U$ have degree $d(u)\leq \delta$ as for
  nodes $v$ in $U$ of degree $d(v)<\delta$, the condition on the
  neighboring colors is trivial (we then have
  $\lfloor\lambda\cdot d(v)\rfloor=0$) and we can thus remove
  such nodes from $U$.

  We transform the graph $B$ into a bipartite graph
  $B'=(U' \dot{\cup} V, E')$ as follows. Each node $v\in U$
  arbitrarily partitions its $d(v)\geq \delta$ neighbors $N(v)$ into
  parts $N_1(v),\dots,N_{k_v}(v)$ of size $\delta/2<N_i(v)\leq \delta$.
  Note that such a partition is always possible. If $N(v)$ is
  partitioned into $k_v$ parts, node $v$ is replaced by $k_v$ nodes
  $v_1,\dots,v_{k_v}$ in $U'$, where node $v_i$ is connected to
  $N_i(v)$. Note that when running a distributed algorithm on $B'$,
  node $v$ can simulate all nodes $v_1,\dots,v_{k_v}$. Clearly in
  $B'$, all nodes $U'$ have a degree in $(\delta/2,\delta]$. We can
  therefore run the weak local splitting oracle on $B'$ and get a
  coloring of $V$ such that each node in $U'$ has at least one red
  neighbor and at least one blue neighbor in $V$. This implies that
  each node $v\in U$ has at least $k_v\geq d(v)/\delta$ neighbors of
  each color in $V$ and we have therefore solved the $\lambda$-local
  splitting problem on $B$. The cost of the reduction is $O(1)$.
\end{proof}

We next prove that the $\lambda$-local splitting problem is
\polyseqloc-hard for any $\lambda=1/\poly\log n$.

\begin{lemma}\label{lemma:localsplithardness}
  For any $\lambda= 1/\log^{O(1)} n$, the problem of computing a
  $\lambda$-local split of an $n$-node bipartite $B=(U \dot{\cup} V,
  E)$ is \polyseqloc-hard.
\end{lemma}
\begin{proof}
  We reduce the the problem of computing a conflict-free
  multicoloring of a given hypergraph $H$ to the given local
  splitting problem. Hence, assume that we are given an $n$-node
  hypergraph $H=(V,E)$ with at most polynomially many hyperedges for which
  we want to compute a conflict-free multicoloring with $\poly\log n$
  colors by using a $\lambda$-local splitting oracle.

  The reduction consists of $\ell$ phases similar to the algorithm
  described in the proof of \Cref{lemma:lowrank_cfcoloring}. In each
  phase, we remove some of the hyperedges and some of the nodes. Let
  $H_i=(V_i,E_i)$ be the hypergraph before starting phase $i$, i.e.,
  we have $H_1=H$. In each phase $i$, we color the nodes with colors
  from a new set $C_i$ of $q=\log^{O(1)} n$ colors and we afterwards
  remove all hyperedges from $H_i$ which contain exactly one node with
  color $x$ for some $x\in C_i$. This will guarantee that all
  remaining hyperedges are large and we can then use the local
  splitting oracle to also remove some nodes. As in
  \Cref{lemma:lowrank_cfcoloring}, the process ends when all
  hyperedges are removed. The goal of each phase is to reduce the rank
  of the hypergraph by a factor at least $1-\lambda/2$. Let $R_i$ be
  the maximum hyperedge size (i.e., the rank) of $H_i$. Note that we
  have $R_1\leq n$. We thus need to show that for all $i$,
  $R_{i+1}\leq (1-1/\lambda)R_i$. Note that this implies that the
  reduction requires
  $\ell O\big(\frac{\log n}{\lambda}\big)=\log^{O(1)} n$ phases and we
  thus compute a conflict-free multicoloring of $H$ with at most
  $q\ell=\poly\log n$ colors as required.

  Let us now consider a single phase $i$ of the reduction. We define
  $\delta:=1/\lambda$ and we define the set
  \mbox{$L_i := \set{e\in E_i : |e| \leq \delta}$} to be the set hyperedges of graph
  $H_i$ of size at most $\delta$. Let $H_i'=(V,L_i)$ be the
  sub-hypergraph of $H_i$ which only contains the edges in $L_i$. We
  then compute a $q$-color conflict-free multicoloring of $H_i'$ for
  some $q=\log^{O(1)} n$. Lemma \ref{lemma:lowrank_cfcoloring}
  guarantees that we can do this deterministically in $\poly\log n$
  rounds in the \LOCAL model. This makes sure that for all hyperedges
  $e$ in $L_i$, there is a color $x$ so that exactly one node in $e$
  has color $x$. Note that if $R_i\leq \delta$, all hyperedges of
  $H_i$ are in $L_i$ and are therefore done. In the following, we thus
  assume that $R_i>\delta$. After removing all hyperedges in $L_i$,
  the resulting graph $H_i''=(V_i,E_i\setminus L_i)$ has only
  hyperedges of size larger than $\delta$ and it remains to compute a
  conflict-free multicoloring of $H_i''$.

  Let $B_i:=(U_i\dot{\cup} V_i, E_{B_i})$ be the bipartite graph which
  is obtained from $H_i''$ in the following natural way. The left side $U_i$
  contains a node for every hyperedge of $H_i''$, whereas the right
  side consists of the nodes $V_i$ of $H_i''$. The node $u\in U_i$
  corresponding to some hyperedge $e\in E_i\setminus L_i$ is connected
  to all the nodes in $V_i$ which are contained in $e$. In the
  following, let $d_i(u)$ be the degree of a node $u\in U_i$ in the
  bipartite graph $B_i$. Note that for all $u\in U_i$, we have
  $d_i(u)> \delta$. Using the $\lambda$-local splitting oracle, we now
  compute a $\lambda$-local splitting of the bipartite graph
  $B_i$. Note that because we assumed that $H$ has only polynomially
  many hyperedges, the bipartite graph $B_i$ also has at most
  polynomially many nodes and we can therefore efficiently simulate
  graph $B_i$ on the network graph $H$. This assigns colors red and
  blue to the nodes in $V_i$ such that every node $u\in U_i$ has at
  least
  $\lfloor\lambda d_i(u)\rfloor=\lfloor d_i(u)/\delta\rfloor \geq 1$
  neighbors of each color. Let $V_{i,R}$ be the set of red nodes. We
  define $H_{i+1}$ to be the sub-hypergraph of $H_i''$ which is
  induced by only the red nodes $V_{i,R}$. That is, for each hyperedge
  of $e$ of $H_i''$, the hypergraph $H_{i+1}$ contains a hyperedge
  consisting of the nodes $e\cap V_{i,R}$. Because each node in the
  bipartite graph $B_i$ has at least one red neighbor, these
  hyperedges are non-empty and therefore a conflict-free
  multicoloring of $H_{i+1}$ directly implies a conflict-free
  multicoloring of $H_i''$ (by potentially adding one additional
  color to the blue nodes to make sure that every node has at least
  one color). Because in $B_i$, every node $u\in U_i$ has at least
  $\lfloor\lambda d_i(u)\rfloor$ blue neighbors, the maximum hyperedge
  size of $H_{i+1}$ is upper bounded by
  \[
    R_{i+1}  \leq  R_i -\lfloor \lambda R_i\rfloor 
    = R_i -\left\lfloor\frac{R_i}{\delta}\right\rfloor
    \stackrel{(R_i>\delta)}{\leq}
    \left(1-\frac{1}{2\delta}\right)\cdot R_i.
  \]
  This concludes the proof. 
\end{proof}

\subsection{Local Splitting is in $\polyseqloc$}

We next present a deterministic algorithm in the \SLOCAL model with
locality $\poly\log n$ that solves the $\lambda$-local splitting
problem on bipartite graphs $B=(U\dot{\cup} V, E)$, where the minimum
degree of nodes in $U$ is $\Omega(\log^2 n)$. While the problem is
shown to be $\polyseqloc$-hard even for $\lambda=1/\poly\log n$ in
\Cref{lemma:localsplithardness}, our \SLOCAL algorithm achieves a much
better split and even works for values of $\lambda$ which are close to
$1/2$. Our algorithm also directly shows that for the given graphs,
the weak local splitting problem is in \SLOCAL.

\begin{lemma}\label{lemma:seqlocalsplit}
  Let $c$ and $d$ be sufficiently large positive constants.  Then, for
  the family of $n$-node bipartite graphs $B=(U\dot{\cup}V,E$ where
  every node in $U$ has degree at least $\delta\geq c\ln^2n$, the
  $\lambda$-local splitting problem is in \polyseqloc\ for any
  $\lambda\leq 1/2 - d\cdot\frac{\ln n}{\sqrt{\delta}}$.
\end{lemma}
\begin{proof}
  We saw in \Cref{lemma:phaseReduction} that we can transfer any
  deterministic $k$-phase \SLOCAL algorithm into a deterministic
  single-phase \SLOCAL algorithm, while incurring only a $k\log^2 n$
  factor increase in the complexity. Leveraging this point, here we
  provide a $2$-phase algorithm $\calA$ where each phase has
  locality no more than $O(\poly\log n)$.

  Let us assume that we are given a bipartite graph
  $B=(U\dot{\cup}V, E_B)$, where every node $u$ in $U$ has degree
  $d(u)\geq \delta\geq c\ln^2n$. Based on graph $B$, we define a graph
  $G=(V,E_G)$ which contains a node for each ``right-side'' node $v\in
  V$ of $B$. Two nodes $\set{v,v'}\in V$ are connected by an edge in
  $G$ if and only if $v$ and $v'$ have a common neighbor in $U$ in
  graph $B$.

  In the first phase of Algorithms $\calA$, we compute a
  $(O(\log n),O(\log n))$-decomposition of the graph $G$. Such a
  decomposition can be computed with $\poly\log n$ locality by
  \Cref{lemma:networkDecomp}. Recall that this partitions the nodes $V$ into
  clusters which are colored with $O(\log n)$ colors. Because two
  clusters with the same color cannot be neighbors, for every cluster
  color, every node $u\in U$ on the ``left side'' of the bipartite graph
  $B$ can only have neighbors from one cluster per color. For every node $u\in
  U$, the set of neighbors $N(u)$ therefore belong to at most $O(\log
  n)$ different clusters.

  For each cluster $\calC\subseteq V$, we now consider the induced
  bipartite graph $B_\calC=(U_{\calC}\dot{\cup}V_\calC,E_{\calC})$ consisting of
  all nodes $V_\calC=\calC$ and all nodes $U_{\calC}$ in $U$ which
  have at least one neighbor in $\calC$. In an internal computation
  within each cluster, the nodes of all clusters are colored
  independently. Note that such a computation within a cluster can be
  done with locality $O(\log n)$ as each cluster has diameter
  $O(\log n)$.

  To see how each cluster is colored, we consider the properties of a
  random coloring of a cluster $\calC$, where each node in $\calC$ is
  independently colored red or blue with probability $1/2$. Let
  $u\in U_{\calC}$ be a node of $B_{\calC}$ and let $N_{\calC}(u)$ be
  the neighbors of $u$ in $B_{\calC}$. By applying a standard Chernoff
  bound, with probability $1-1/(2n)$, the absolute difference between
  the number of red and blue nodes in $N_{\calC}(u)$ can be upper
  bounded by a term of order
  $O(\sqrt{|N_{\calC}(u)|\log n} + \log n)$. A union bound over all
  nodes in $U_{\calC}$ implies that there exists a red/blue coloring
  of the nodes in $\calC$ such that for all $u\in U_{\calC}$, the
  absolute difference in the number of red and blue nodes among the
  nodes in $N_{\calC}(u)$ is at most
  $\alpha(\sqrt{|N_{\calC}(u)|\log n}+ \log n)$ for some constant
  $\alpha>0$. In one \SLOCAL-phase with locality $O(\log n)$, such a
  red/blue-coloring can be computed for every cluster.

  Recall that for each node $u\in U$, the neighborhood $N(u)$ is
  partitioned among at most $O(\log n)$ different clusters. Assume
  that the set $N(u)$ is partitioned among $k_u$ clusters and that it
  is partitioned into sets of sizes $n_{u,1},\dots,n_{u,k_u}$. By
  combining the red/blue-colorings of all the clusters, we therefore
  obtain a red/blue-coloring of the whole set $V$ such that for every
  node $u\in U$, the absolute difference between the number of red and
  blue nodes in $N(u)$ is upper bounded by
  \begin{eqnarray*}
    \alpha\cdot\sum_{i=1}^{k_u} \big(\sqrt{n_{u,i}\ln n} + \ln n\big)
    & \leq &
             \alpha\cdot\left(\sqrt{k_u\sum_{i=1}^{k_u}n_{u,i}\ln n} +
             k_u\ln n\right)\\
    & = & \alpha\cdot\left(\sqrt{k_u |N(u)|\ln n} + k_u\ln n\right).
  \end{eqnarray*}
  The inequality in the first line follows because for any integer $k\geq 1$ and any
  $x_1,\dots,x_k\geq 0$, it holds that
  $\sum_{i=1}^{k}\sqrt{x_i}\leq \sqrt{k}\sqrt{\sum_i^kx_i}$, which
  follows from the Cauchy-Schwarz inequality. The claim of the lemma
  now follows by using that $k_u=O(\log n)$ and by setting the constants
  $c$ and $d$ in the lemma statement large enough.  
\end{proof}


\section{Approximating Covering and Packing Integer Linear Programs}
\label{sec:ILP}

In this section, we explain \SLOCAL algorithms with complexity $O(\poly(\log n/\eps))$ for computing $(1+\eps)$ approximations of covering and packing Integer Linear Programs (ILP). In conjunction with \Cref{thm:randomizedequality}, this implies that the same approximation can be achieved using randomized \LOCAL algorithms with complexity $O(\poly(\log n/\eps))$. Furthermore, if one can deterministically solve one of the problems shown to be \polyseqloc-complete in the previous sections---for instance, local splitting, hypergraph conflict-free multi-coloring, or network decomposition---in $\poly\log n$ rounds of the \LOCAL model, then we would get $O(\poly(\log n/\eps))$ round deterministic algorithms in the \LOCAL model for $(1+\eps)$ approximation of covering and packing ILPs.\\

The formulation of covering and packing ILPs, which are duals of each other, is as follows:

\begin{figure*}[h!]
\centering
\begin{minipage}{0.43\textwidth}
\begin{mdframed}[hidealllines=true, backgroundcolor=gray!20]
\vspace{-0.5cm}
\begin{eqnarray*}
\textit{\textbf{Covering ILP:}}\\
\min &\mathbf{c}^{T} \mathbf{x}\\ 
\textit{subject to } &A\mathbf{x}\geq \mathbf{b} \\ 
&\mathbf{x}\in \mathbb{N}_{\geq 0}^{n_1} \\
&A\geq 0, \mathbf{c}\geq 0, \mathbf{b}\geq 0 
\end{eqnarray*}
\end{mdframed}
\end{minipage}
\hspace{1cm}
\begin{minipage}{0.43\textwidth}
\begin{mdframed}[hidealllines=true, backgroundcolor=gray!20]
\vspace{-0.5cm}
\begin{eqnarray*}
\textit{\textbf{Packing ILP:}}\\
\max &\mathbf{b}^{T}\mathbf{y}\\ 
\textit{subject to } &A^{T}\mathbf{y}\leq \mathbf{c} \\ 
&\mathbf{y}\in \mathbb{N}_{\geq 0}^{n_2} \\
&A\geq 0, \mathbf{c}\geq 0, \mathbf{b}\geq 0 
\end{eqnarray*}
\end{mdframed}
\end{minipage}
\end{figure*}

\noindent We imagine these LPs are represented via bipartite graphs $H=(V, E)$, where $V=L\cup R$ and $E\subseteq L\times R$ as in \cite{papa93,bartal97,nearsighted}. There is one vertex $\ell\in L$, $|L|=n_1$, representing each variable and one vertex $r\in R$, $|R|=n_2$, representing each linear constraint. The edges of the bipartite graph are such that each variable vertex related to $x_i$ (or $y_j$) is connected to all linear constraint vertices that have a non-zero coefficient for $x_i$ (respectively $y_j$). Various classic optimization problems can be easily viewed in this framework, with no more than an $O(1)$ factor loss in the locality. This includes covering ILPs such as minimum dominating set, set cover, and vertex cover and packing ILPs such as maximum independent set and maximum matching. For instance, for maximum independent set in a graph $G=(V, E)$, we have one variable vertex per each node of $G$, and one constraint vertex per each edge $e=(v, u)\in E$, which can be simulated by one of its endpoints, say the one with the larger ID.

In the following, we provide simple deterministic \SLOCAL\ algorithms with locality $O(\poly(\frac{\log n}{\eps}))$ for covering and packing ILPs. For simplicity, instead of presenting the algorithms in the general framework, we explain the algorithms for two concrete sample problems, maximum independent set and minimum dominating set. It is easy to see how these algorithms can be extended to the related general cases of packing and covering ILPs, respectively. The resulting time complexity will be polylogarithmic in $n$, $1/\eps$, and in the ratio between the largest and smallest weight and coefficient.

\subsection{Sample Packing Problem: Approximating Maximum Independent Set}
\label{sec:toolAlgorithmDesign}
\begin{theorem}
There is a deterministic algorithm with complexity $O(\poly(\log n/\eps))$ in the \SLOCAL\ model that computes a $(1+\eps)$-approximation of the maximum independent set.
\end{theorem}

\begin{proof}We use a simple ball growing argument. Suppose that $v_1$, $v_2$, \dots, $v_n$ is the ordering of the vertices provided to the \SLOCAL algorithm. 

Let $\alpha(H)$ denote the independence number of graph $H$, i.e., its maximum independent set size. We begin with an empty global independent set. We start with some node $v_1$ and find a radius $r$ such that $\alpha(G[B_{r+1}(v)]) \leq (1+\epsilon)\cdot\alpha(G[B_{r}(v)])$. Notice that $r\leq R=\bigO(\log n/\eps)$. Compute a maximum independent set of $B_{r}(v)$, add it to the global independent set, and remove $B_{r+1}(v)$ from the graph. This clearly has locality $O(\log n/\epsilon)$. Furthermore, it provides a $(1+\epsilon)$ approximation of the maximum independent set. The reason is as follows: we can decompose the optimal maximum independent set $I^*$ into $n$ (potentially empty) subsets $I_1$, \dots, $I_n$, each being the vertices of $I^*$ which are removed when processing node $v_i$. Then, the computed independent set when processing $v_i$ has size at least $|I_i|/(1+\eps)$. Thus, overall, the computed independent set has size at least $|I^*|/(1+\eps)$.
\end{proof}
\begin{corollary}There is a randomized algorithm with complexity $O(\poly(\log n/\eps))$ in the \LOCAL\ model that computes a $(1+\eps)$-approximation of the maximum independent set, with high probability.
\end{corollary}

We remark that, to the best of our knowledge, this is the first algorithm providing this high probability approximation for maximum independent set. Prior to our work, it was known how to randomly compute an independent set whose size is \emph{in expectation} a $(1+\epsilon)$ approximation of maximum independent set\cite{podc16_BA}. However, we are not aware of a method for transforming that algorithm to a high probability approximation guarantee, and indeed, due to the nature of the \LOCAL model, such a transformation does not seem feasible, or at least is not straightforward.

\subsection{Sample Covering Problem: Approximating Minimum Dominating Set}
\label{sec:toolAlgorithmDesign2}
\begin{theorem}
There is a deterministic algorithm with complexity $O(\poly(\log n/\eps))$ in the \SLOCAL\ model that computes a $(1+\eps)$-approximation of the minimum dominating set.
\end{theorem}
\begin{proof} Again, we use a simple ball growing argument. Suppose that $v_1$, $v_2$, \dots, $v_n$ is the ordering of the vertices provided to the \SLOCAL algorithm. 

For a node $v$, let $g(v, r)$ be the size of the smallest set of vertices in $B_{r+1}(v)$ that dominates $B_{r}(v)$. We begin with an empty global dominating set. We start with some node $v_1$ and find a radius $r$ such that $g(v, r+2) \leq (1+\epsilon)\cdot g(v,r)$. Notice that $r\leq R=\bigO(\log n/\eps)$. Compute a smallest set in $B_{r+3}(v)$ that dominates $B_{r+2}(v)$, add it to the global dominating set, and remove $B_{r+2}(v)$ from the graph. Call $B_{r}(v)$ the \emph{central ball} of this step. This clearly has locality $O(\log n/\epsilon)$. Furthermore, it provides a $(1+\epsilon)$ approximation of the minimum dominating set. The reason is as follows: construct node sets $V_1$, $V_2$, \dots, $V_n$ and add each vertex $v\in V$ to the subset $V_i$ such that $v$ was in the \emph{central ball} $B_{r}(v_i)$ when processing vertex $v_i$. Notice that some vertices $v$ will be in none of the sets $V_i$. On the other hand, each two sets $V_i$ and $V_j$ have distance at least $3$. Hence, no node can dominate vertices from two or more of these sets. Consider the optimal minimum dominating set $D^*$ and partition it into $n$ disjoint (potentially empty)  subsets $D_1$, \dots, $D_n$, each being the set of vertices of $D^*$ that dominate $V_i$. Then, the computed dominating set when processing $v_i$ has size at most $|D_i|(1+\eps)$. Thus, overall, the computed dominating set has size at most $|D^*|(1+\eps)$.
\end{proof}


\section{On The Power of the Sequential LOCAL Model}\label{appsec:inPSLOCAL}
\label{sec:powerOfSequentialModel}

As mentioned before,  the \SLOCAL model is quite powerful, thanks to the fact that vertices are processed in a sequential order and that each vertex $v$ has a local state $S_v$ to record the information it gathered. Because of this, the model is clearly stronger than the standard \LOCAL model. In fact, a priori, the \SLOCAL model might look too strong to be of any interest: in particular, it can easily solve all the classic problems of interest---e.g., maximal independent set, $(\Delta+1)$-vertex coloring, $(2\Delta-1)$-edge coloring, or maximal matching---with locality just $O(1)$. 

In this section, we show that, perhaps surprisingly, the (randomized) \SLOCAL model is not much more powerful than the randomized \LOCAL\ model, when we are concerned with polylogarithmic locality. Furthermore, as we prove in \Cref{lemma:phaseReduction}, even if we allow the \SLOCAL algorithm to use a polylogarithmic number of phases and process the vertices sequentially for a polylogarithmic number of iterations, the power does not change significantly.

\subsection{Random Sequential vs. Random Distributed Local Algorithms}

\noindent\textbf{\Cref{thm:randomizedequality} (restated).} \emph{$\polyrandseqloc{\eps} \subseteq \polyrandloc{\eps+1/poly(n)}.$}

\begin{proof}
Given a randomized \SLOCAL algorithm $\mathcal{A} \in \polyrandseqloc{\eps}$ with locality $r=\poly\log n$, we explain a randomized \LOCAL algorithm $\mathcal{B}\in \polyrandloc{\eps+1/poly(n)}$ with locality $\poly\log n$ that simulates $\mathcal{A}$. The first step in algorithm $\mathcal{B}$ is to compute an $(\bigO(\log n), \bigO( \log n))$-network decomposition of the graph $G^{r+1}$, using the randomized algorithm of Linial and Saks\cite{linial93} in $O(r\log^2 n)$ time. This network decomposition partitions the vertices of $G$ into clusters $X_1$, $X_2$, \dots, $X_\eta$ such that it satisfies the following two properties with probability at least $1-1/\poly(n)$: 
\begin{itemize}
\item[(1)] any two vertices of each cluster have distance at most $\bigO(r\log n)$ in $G$, and 
\item[(2)] each cluster $X_i$ is assigned a color in a color set $\{1, 2, \dots, Q\}$ for a $Q=\bigO(\log n)$ such that any two clusters of the same color have distance at least $r+1$ in $G$.
\end{itemize}
	
To simulate the \SLOCAL algorithm $\mathcal{A}$, we use this network decomposition to generate an ordering $\pi$ of vertices as in \Cref{observation:lowDiameterNetworkDecomp}
, this will be the order on which we assume $\mathcal{A}$ operates. 

The algorithm $\mathcal{B}$ now uses this order $\pi$ to simulate $\mathcal{A}$. Algorithm  $\mathcal{B}$ works in $Q=\bigO(\log n)$ phases, each taking $\bigO(r\log n)$ rounds. In the $i^{th}$ phase, each vertex $v_{\ell}$ in a cluster $X_j$ with color $i$ first gathers all the information in the $r$-neighborhood of the cluster $X_j$. Then, node $i$ locally simulates the algorithm $\mathcal{A}$ for all the nodes in $X_j$, according to the order $\pi$. For each node $u$ in $X_j$, to determine the output of $u$, the simulation will need to know the state $S_w$ of nodes $w$ which appear before $u$ and are within distance $r$ of $u$. If $w$ has color $i'< i$, this state is written in the local memory of $S_w$ when simulating phase $i'$ and thus $u$ knows it, as it has gathered the information in the $r$-hop neighborhood of $X_i$. If $w$ has color $i$, then node $u$ simulated node $w$ before and thus knows $S_w$. Notice that nodes of different clusters of the same color $i$ can perform this process in parallel as their computations do not influence each other (because of the way $\pi$ is defined). 
\end{proof}

The lemma easily generalizes to show that $\randseqloc{\eps}(t^{\bigO(1)}(n)) \subseteq \randloc{\eps+1/poly(n)}(t^{\bigO(1)}(n))$, for any function $t(n)\geq \log n$. 

\subsection{Multi-Phase versus Single-Phase Sequential Local Algorithms}
We call \SLOCAL\ algorithms as defined in \Cref{sec:SLOCAL} \emph{single-phase} \SLOCAL algorithms because they process each node only once. If we allow an algorithm to run through the nodes $k$ times, we call it a \emph{$k$-phase} \SLOCAL algorithm. We next prove that having multiple phases does not increase the power significantly. In particular, the set of problems which can be solved with polylogarithmic locality in the \SLOCAL\ model does not change if we allow $k$ phases as long as $k$ is polylogarithmic.

\noindent{\Cref{lemma:phaseReduction} (restated).
\emph{Any  $k$-phase \SLOCAL\ algorithm $\mathcal{A}$ with locality $r_i$ in phase $i=1,\ldots,k$ can be transformed into a single-phase \SLOCAL\ algorithm $\mathcal{B}$ with locality $r_1+2\sum_{i=2}^kr_i$.}
\medskip
\begin{proof}
We prove that a $k$-phase algorithm $\mathcal{A}$ with locality $r_i$ in phase $i$ can be transferred into a single phase algorithm $\mathcal{B}$ with locality $R:=\sum_{i=1}^k r_i$ if we assume that node $u$ in algorithm $\mathcal{B}$ can write into the memory of nodes in $B_{R-r_1}(u)$. Then the claim follows with \Cref{observation:memoryWriting}.
\smallskip

We explain how to transform a two phase \SLOCAL\ algorithm $\mathcal{A}'$ with locality $r_1$ in the first phase and $r_2$ in the second phase  into a single phase \SLOCAL\ algorithm $\mathcal{B}'$ with locality $r_1+r_2$. Then the aforementioned transformation of $\mathcal{A}$ into $\mathcal{B}$ can be deduced with an inductive argument.

To construct algorithm $\mathcal{B'}$ we need to see that the output in phase two of node $u$ in algorithm $\mathcal{A}'$ only depends on the output of the first phase of all nodes in $B_{r_2}(u)$ and the output of the second phase of nodes in $B_{r_2}(u)$ that have  been processed in the second phase before $u$. 

\noindent\paragraph{Algorithm {\boldmath$\mathcal{B}'$}:} Assume nodes in $\mathcal{B}'$ are processed according to order $\pi$. Whenever it is $u$'s turn in $\mathcal{B}'$, it collects its neighborhood $B_{r_1+r_2}(u)$, $u$ simulates the first phase of algorithm $\mathcal{A}'$ for all nodes in $B_{r_2}(u)$ and writes the output into the memories of the nodes in $B_{r_2}(u)$. In this simulation $u$ takes into account that some nodes in this ball might already have computed their output because they were handled before $u$ or because some other node wrote their output into their memory. In particular, all nodes which are processed before $u$ in order $\pi$ have already computed their output for phase two. Note that this simulation might use different orders for the two phases of $\mathcal{A}'$.

Then $u$ has all the information to compute its output after two phases, i.e., the phase one output and memory content of nodes in $B_{r_2}(u)$ and the phase two output of nodes in $B_{r_2}(u)$ of the nodes that are ordered before $u$ in $\pi$. 
\end{proof}

\section{Low Diameter Ordering \& Network Decomposition are in {\boldmath$\polyseqloc$}}
\subsection{Network Decomposition via Sequential Ball Growing}
\label{alg:APLS}

In this section, we review the centralized sequential $\big(O(\log n), O(\log n)\big)$-decomposition algorithm, which is contributed to Linial and Saks\cite{linial93} and Awerbuch and Peleg\cite{Awerbuch-Peleg1990}. 

Recall from \Cref{def:decomposition} that a weak \emph{$\big(d(n),c(n)\big)$-decomposition} of an
  $n$-node graph $G=(V,E)$ is a partition of $V$ into clusters such that each cluster has weak diameter at most $d(n)$ and the
  cluster graph is properly colored with colors $1,\dots,c(n)$. We refer to the vertices of the clusters of each color $i$ as block $i$ and denote them by $V_i$. Thus, this decomposition partitions $V$ into \emph{blocks} $V_1,\ldots,V_{c(n)}$. 

The sequential algorithm of \cite{linial93, Awerbuch-Peleg1990} constructs the decomposition one block at a time. We describe one block of the construction, show that it produces non-adjacent clusters each with weak diameter $d(n)=O(\log n)$, and argue that it removes a constant fraction of the nodes. Thus, after $O(\log n)$ blocks, all nodes are removed and thus we have a $\big(O(\log n), O(\log n)\big)$-decomposition.

\paragraph{Construction of one block:} Let $G_i=G[V\setminus \big(V_1\cup\ldots \cup V_{i-1}\big)]$ be the subgraph of $G$ left after removing the vertices of blocks $V_1$ to $V_{i-1}$. We construct the clusters of the block $V_i$, one at a time. During this process, we will discard some vertices of $G_i$, once they are processed, and thus $G_i$ is gradually shrinking.

Repeat the following process until $G_i$ is empty: Pick an arbitrary vertex $v \in G_i$ and start the following ball growing process on $G_i$: Find the smallest radius $r^*$ such that
\begin{align}
\label{eqn:ballGrowing}
\frac{|B^{G_i}_{r^*+1}(v)|}{|B^{G_i}_{r^*}(v)|}\leq 2.
\end{align}
Note that $r^*\leq \log_2 n$, because otherwise we would have $|B^{G_i}_{\log_2 n+1}(v)| >n$, which would be a contradiction with the graph having only $n$ vertices. Add nodes of $B^{G_i}_{r^*}(v)$ as one cluster of $V_i$, and then remove nodes $B^{G_i}_{r^*+1}(v)$ from $G_i$.

\begin{lemma}The sequential ball growing algorithm of Linial and Saks\cite{linial93} and Awerbuch and Peleg\cite{Awerbuch-Peleg1990} described above computes an $\big(O(\log n), O(\log n)\big)$-decomposition. 
\end{lemma}
\begin{proof}
It is easy to see that due to condition \Cref{eqn:ballGrowing}, each block removes at least a constant fraction of the unclustered nodes. Hence, $O(\log n)$ blocks suffice. 

In each block $V_i$, each cluster has weak diameter at most $2r^*\leq 2\log_2 n$, because it was found as a ball of radius at most $r^*$ around some node $v$. Furthermore, no two clusters of the same block are adjacent because when constructing the first cluster, its boundary nodes are removed from the graph but not added to the cluster, that is, we remove $B^{G_i}_{r^*+1}(v)$ but define only $B^{G_i}_{r^*}(v)$ to be a cluster.
\end{proof}

\subsection{Low Diameter Ordering \& Network Decomposition are in {\boldmath$\polyseqloc$}}
Now, we adapt the deterministic sequential algorithm of the previous subsection to the \SLOCAL model. This allows us to compute a network decomposition, and also a low-diameter ordering, in $\poly\log n$ rounds of the \SLOCAL model.
\begin{lemma}
\label{lemma:networkDecomp}
Computing a $\big(O(\log n), O(\log n)\big)$-decomposition of a given $n$-node graph is in $\polyseqloc$.
\end{lemma}
\begin{proof}
The proof of \Cref{lemma:phaseReduction} shows how a $\big(\bigO(\log n), \bigO(\log n)\big)$-decomposition  can be computed in a single phase.
%
\end{proof}
Alternatively to the above proof and if one assumes that nodes can write into other nodes' memory (cf. \Cref{observation:memoryWriting}), the deterministic sequential $\big(\bigO(\log n), \bigO(\log n)\big)$-decomposition algorithm from the previous section directly translates into an  \SLOCAL algorithm  with $\bigO(\log n)$ phases, which then can be transferred into a single-phase \SLOCAL algorithm with polylogarithmic locality with \Cref{lemma:phaseReduction}.

\begin{lemma} 
  The problem of computing a $\poly\log n$-diameter ordering is in \polyseqloc.
\end{lemma} 
\begin{proof} 
  The result follows with \Cref{lemma:networkDecomp} and
  \Cref{observation:lowDiameterNetworkDecomp}.
\end{proof}


\newpage
\bibliographystyle{alpha}
\bibliography{references}

\end{document}